\pdfoutput=1
\PassOptionsToPackage{breaklinks}{hyperref}
\documentclass{eptcs}

\usepackage{amsmath} 
\usepackage{amsthm} 
\usepackage{amssymb} 

\usepackage{pst-pdf}
\usepackage{tikz} 

\usepackage{enumerate} 
\usepackage{longtable} 

\usepackage[firstinits=true,maxcitenames=5]{biblatex} 
\usepackage{cleveref} 
\DeclareFieldFormat[article]{volume}{\textbf{#1}}
\DeclareFieldFormat{pages}{#1}
\renewbibmacro{in:}{%
  \ifentrytype{article}{}{\printtext{\bibstring{in}\intitlepunct}}}
\defbibheading{bibliography}{\section*{#1}}

\newtheorem{theorem}{Theorem}
\newtheorem{thm}[theorem]{Theorem}
\theoremstyle{definition}
\newtheorem{definition}[theorem]{Definition}
\newtheorem{example}[theorem]{Example}
\newtheorem{cor}[theorem]{Corollary}
\newtheorem{lemma}[theorem]{Lemma}
\newtheorem{prop}[theorem]{Proposition}
\newcommand{\beq}{\begin{equation}}
\newcommand{\eeq}{\end{equation}}
\crefname{trsf}{transformation}{transformations}

\usetikzlibrary{shapes,positioning,arrows}
\tikzset{>=stealth', c/.style={draw, regular polygon, regular polygon sides=3, minimum height=2.4em, inner sep=0}, q/.style={draw, circle, minimum height=1.8em, inner sep=0}, e/.style={<-, thick}}
\tikzset{chi/.style={c, fill=yellow!25}}

\newcommand{\Hb}{\mathcal{H}}

\DeclareMathAlphabet{\amscal}{OMS}{cmsy}{m}{n}
\newcommand{\Cs}{\ensuremath{\amscal{C}}}
\newcommand{\Qs}{\ensuremath{\amscal{Q}}}
\newcommand{\Gs}{\ensuremath{\amscal{G}}}
\newcommand{\Is}{\ensuremath{\amscal{I}}}

\newcommand\indep{\protect\mathpalette{\protect\independenT}{\perp}}
\def\independenT#1#2{\mathrel{\rlap{$#1#2$}\mkern2mu{#1#2}}}
\newcommand{\mi}[2]{#1 \indep #2}
\newcommand{\ci}[3]{#1 \indep #2\ |\ #3}
\newcommand{\dsep}[3]{#1 \bot #2\ | \ #3}

\newcommand{\mutual}[2]{\operatorname{I}(#1:#2)}
\newcommand{\cmutual}[3]{\operatorname{I}(#1:#2|#3)}
\newcommand{\shan}[1]{\operatorname{H}(#1)}
\newcommand{\cshan}[2]{\operatorname{H}(#1|#2)}

\newcommand{\qpath}{\rightsquigarrow}

\newcommand{\ich}[1]{m(#1)}
\newcommand{\ian}[1]{J^{-}(#1)}
\newcommand{\pa}[1]{\text{pa} \, {#1}}

\newcommand{\PA}[1]{\text{PA} \, {#1}}

\newcommand{\cpa}[1]{\text{opa} \, {#1}}

\newcommand{\GPA}[1]{\text{incU} \, {#1}}
\newcommand{\GCH}[1]{\text{outU} \, {#1}}
\newcommand{\s}[1]{^{(#1)} }

\newcommand{\cT}{\mathcal{T}}

\DeclareMathOperator{\Tr}{Tr}

\newcounter{dagcounter}
\newcommand{\dagc}{\#\refstepcounter{dagcounter}\thedagcounter}

\title{Theory-independent limits on correlations from \\ generalised Bayesian networks}
\author{\begin{tabular}{ccc}
Joe Henson & Raymond Lal & Matthew F. Pusey\\
{\footnotesize Imperial College London and} & {\footnotesize University of Oxford and} & {\footnotesize Perimeter Institute}\\
{\footnotesize University of Bristol} & {\footnotesize University of Cambridge}\\
{\footnotesize {\tt j.henson@bristol.ac.uk}} & {\footnotesize \tt rayl@cs.ox.ac.uk} & {\footnotesize\tt m@physics.org}\\
\end{tabular}}

\def\titlerunning{Theory-independent limits on correlations from generalised Bayesian networks}
\def\authorrunning{J. Henson, R. Lal \& M. F. Pusey}

\begin{document}
\maketitle

\begin{abstract}
Bayesian networks provide a powerful tool for reasoning about probabilistic causation, used in many areas of science.
They are, however, intrinsically classical.
In particular, Bayesian networks naturally yield the Bell inequalities.
Inspired by this connection, we generalise the formalism of classical Bayesian networks in order to investigate non-classical correlations in arbitrary causal structures.
Our framework of `generalised Bayesian networks' replaces latent variables with the resources of any generalised probabilistic theory, most importantly quantum theory, but also, for example, Popescu-Rohrlich boxes.
We obtain three main sets of results.
Firstly, we prove that all of the observable conditional independences required by the classical theory also hold in our generalisation; to obtain this, we extend the classical $d$-separation theorem to our setting.
Secondly, we find that the theory-independent constraints on probabilities can go beyond these conditional independences. 
For example we find that no probabilistic theory predicts perfect correlation between three parties using only bipartite common causes.
Finally, we begin a classification of those causal structures, such as the Bell scenario, that may yield a separation between classical, quantum and general-probabilistic correlations.
\end{abstract}

\section{Introduction}\label{sec:intro}
	
Bell's theorem \cite{bell} is a central result in the foundations of quantum mechanics. 
It reveals that certain quantum correlations are stronger than those obtainable in any locally causal model as defined by Bell.
Recently, new results have been obtained by using variations of the scenario that Bell originally considered.
For example, Popescu \cite{pophidden} found that sequences of measurements can reveal nonclassicality in more states than the single measurements considered in a Bell scenario. 
Branciard, Gisin and Pironio \cite{bilocal2} found that including the independence of multiple sources could lead to more robust experiments than the single source assumption of a Bell scenario. 
Using this idea, Fritz \cite{tobias} showed that the `free will' assumption of Bell's theorem can be replaced with an assumption about independence of sources, by replacing the measurement settings of the Bell scenario with additional sources.
Finally, Bancal \emph{et al.}\ \cite{Bancal2012} used an elaborate quadripartite scenario to show that explanations of the violation of Bell inequalities using superluminal but finite speed influences are in conflict with the no-signalling principle.

The common theme in these results is the consideration of more complicated causal structures than the one usually assumed in the Bell scenario.
This leads to new insights into how quantum theory deviates from classical physics: 
by considering arbitrary causal structures, these examples expose a rich structure to quantum correlations.
However, to clarify and unify these results, it would be helpful to have a \emph{general} framework that formalises the connection between causal structure and observable correlations.
There are two desirable features that a general framework of this kind should have.
Firstly, it should describe constraints on locally causal models (i.e.~defined using classical random variables), for arbitrary causal structures, e.g.~it should generalise Bell inequalities.
Secondly, it should also allow for non-classical resources---not only of quantum mechanics, but also those of \emph{generalised probabilistic theories} (GPTs).
The development of GPTs originates in the fact that, in the Bell scenario, quantum theory cannot achieve the strongest correlations that are consistent with the no-signalling principle \cite{Cirelson1980, prbox}.
It would interesting to understand the consequences of different types of causal structure for the separation between classical, quantum, and more general correlations.
In particular, this would allow us to pose the question of what is special about quantum correlations in a wider framework than has so far been used.  

A framework that achieves the first objective is that of \emph{Bayesian networks}, based on directed acyclic graphs (DAGs).
This has been an active area of research by statisticians and computer scientists for several decades, pioneered in particular by Pearl \cite{pearl,pearlintro}. 
When this framework is applied to a Bell-type experiment, and the causal structure implied by special relativity and independence of settings is assumed, one obtains exactly Bell's notion of local causality \cite{wood}. 
The significance of this is two-fold: firstly, Bayesian networks are the natural setting for generalising Bell scenarios; secondly, a new formalism---but structurally similar to Bayesian networks---will be needed to describe the behaviour of quantum theory and other GPTs on arbitrary causal structures. 

\paragraph{Our contribution.}
In this paper, we propose a generalisation of Bayesian networks which incorporates the framework of GPTs.
In particular, we generalise the latent nodes of standard Bayesian networks to allow for resources from an arbitrary GPT.
We then investigate the extent to which results from the causality literature generalise to our approach.
We have three main results.

Our first result shows that all the observable conditional independences that follow from a classical Bayesian network still follow in our generalisation. The conditional independences mandated by a DAG are characterised graphically by the `$d$-separation criterion'. Technically our result is that this criterion is still sound in our generalisation. 
Since our framework goes beyond classical probability theory, we do not have enough structure to even define conditional independences involving latent nodes; hence we require a proof that is very different to the classical case.

Secondly, we also explore what constraints further than the observable conditional independences can be derived for a given causal structure, even in the most general theories.  
In the case of classical Bayesian networks, all constraints on probability distributions implied by the causal structure are (by definition) conditional independences. 
However, these conditional independences may involve `latent' variables, which are unobserved.  Hence not all of the constraints on observable variables need to take the form of \emph{observable} conditional independences. 
For example, in the Bell setup, Bell inequalities are constraints on the observable variables that arise from the existence of latent variables. 
But Bell inequalities are stronger constraints than the observable conditional independences, textit{i.e.}~the no-signalling conditions.

Since our approach will be to allow arbitrary GPTs,  
 the Bell inequalities in the Bell scenario will not constrain the observable probabilities in a general theory.
However, we examine two other quantitative limits on classical correlations that apply to different causal structures.  As in the Bell inequality case, these limits do not follow from the observable conditional independences.
Nevertheless, we find that both limits do carry over to arbitrary GPTs.
Specifically, we show that perfect correlation between three parties cannot be explained by bipartite common causes alone, regardless of which physical theory is used. 
We also show that any GPT obeys the `instrumental inequality', a close cousin of the Bell inequalities that applies to a simple four-node DAG. 

Finally we identify an important classification problem: which are the causal structures that, even classically, have no observable consequences beyond conditional independences? Structures not in this class will certainly be the focus of attention in quantum foundations, but we believe this classification will be of interest in other applications of even the classical causality framework.  We make progress on this problem by providing a sufficient condition for our generalised DAG to imply only the observable conditional independences.

\paragraph{Related work.}
Our work extends Pearl's research programme \cite{pearl} to the study of nonlocality. 
In this respect we build upon the work of Wood and Spekkens \cite{wood}, who showed that such a connection can be made.
Part of our work also builds upon the circuit framework developed by Chiribella, D'Ariano and Perinotti (CDP) \cite{Chiribella2010}.
There are several other lines of investigation with similar but distinct aims to ours.
Leifer and Spekkens have the ambitious aim of an inherently quantum theory of Bayesian networks \cite{neutral}.
However the Leifer-Spekkens approach is work in progress, and is unlikely to allow for other general probabilistic theories.
Fritz has generalised the definitions of classical, quantum, and GPT correlations beyond the Bell scenario, and provided many interesting examples \cite{tobias,tobias2}. But he does not aim to generalise the standard theory of Bayesian networks directly, and so not all of our results can be translated to his definitions.
In \cref{sec:comparison} we discuss the connections to these works in more detail. Related work has meanwhile appeared in \cite{pienaar,infoquant}, the latter including extensions of some of our results.

\paragraph{Plan of paper.}
In Section~\ref{sec:CBN} we introduce the background on classical Bayesian networks, in particular the classical $d$-separation theorem.
In Section~\ref{sec:GBN} we discuss parts of the CDP circuit framework, which we then build upon to define `generalised Bayesian networks'.
We then prove the $d$-separation theorem for our framework.
In Section~\ref{sec:bounds} we investigate bounds on correlations for the triangle and instrumental inequality scenarios. 
Finally, in Section~\ref{sec:towards} we provide a sufficient condition on a causal structure for all sets of correlations to be equal. 

\section{Classical Bayesian networks}\label{sec:CBN}

We often have reasons to assume a given set of causal relations between random variables. 
A basic example is the Bell scenario \cite{bell}, in which we consider probability distributions $P(a,b|x,y)$. 
The underlying spatio-temporal relations are assumed to constrain these distributions by conditional independences known as the `no-signalling' conditions, e.g.~$P(a|x,y)=P(a|x)$. 
Bell's locality condition places a further restriction on the possible correlations:
\beq\label{eq:Bell_locality}
P(a,b|x,y)=\sum_\lambda P(a|x,\lambda)P(b|y,\lambda)P(\lambda).
\eeq
Now, the locality condition can be understood as a condition on the background causal structure, stating that the correlations in $P(a,b|x,y)$ arise through a common cause---a classical random variable $\lambda$---that is in the past of both Alice and Bob.  
Bell inequalities then characterise the correlations that are compatible with this causal structure\footnote{At this point one might question the physical motivations for assuming a particular causal structure, especially with regard to the spatio-temporal causal order that is so crucial to the discussion on Bell's theorem and its consequences.  While much could be said on this issue, the main intention here is to discuss the consequences of assuming a causal structure, rather than the many possible motivations for doing so.}.

In general, how do we characterise the set of allowed probability distributions given a certain causal structure?  
In the case where we only consider causal relations between classical random variables, this question is answered by the theory of Bayesian networks. 
This theory provides a way to describe causal structures, along with rules to determine which probability assignments are consistent with them.  
Here we provide a brief introduction to this aspect of Bayesian networks, with a view to its generalisation in subsequent sections.  
We largely follow Pearl's terminology and notation \cite{pearl}.

\subsection{Probabilities on graphs}

Recall that a \emph{directed graph} $G$ is a pair $(V,E)$, where $V$ is a set of nodes, and $E\subseteq V\times V$ is a set of directed edges.
It is often useful to label the nodes with an index, so that we can write $V=\{X\s{i}\}_i$.
A directed graph may have a \emph{directed cycle}, viz.~a sequence of edges $X\s{1}\to X\s{2}\to\dots\to X\s{n} \to X\s{1}$.
A \emph{directed acyclic graph (DAG)} is a directed graph which has no directed cycles.

In our work, DAGs will represent causal structure: more specifically, an edge $X\to Y$ will represent the possibility of direct causal influence from $X$ to $Y$, where `direct causal influence' will be defined in terms of probabilistic conditional dependence.  
The nodes that can directly influence $Y$ are all nodes $X$ for which there is an edge $X \rightarrow Y$; these are the \emph{parents} of $Y$, and the set of all parents of $Y$ is denoted $\PA{Y}$.  
Similarly, if $X\to Y$ then $Y$ is a \emph{child} of $X$.
A \emph{directed path} is a sequence of nodes $X\s{1},X\s{2},\dots,X\s{n}$ such that $X\s{i} \rightarrow X\s{i+1}$ for $1\leq i \leq n-1$.  
In keeping with familial terminology, we say that $Y$ is a \textit{descendant} of $X$, and $X$ is an \emph{ancestor} of $Y$, if there is a directed path from $X$ to $Y$.  
We also define the following two useful functions on sets of nodes: 
\begin{enumerate}[(i)]
	\item we define $\ich{U}$ to be the union of the set of nodes $U$ with all the children of each of the nodes in $U$;
	\item we define  $\ian{U}$ to be the union of $U$ with the set of all ancestors of nodes in $U$ (the entire `past' of $U$).
\end{enumerate}

Now, consider the Bell scenario in which a common cause is assumed to exist.
The DAG for this scenario is shown in \cref{belldag}, 
\begin{figure}
\begin{center}
\begin{tikzpicture}[baseline=(current bounding box.center)]
  \node[c](A) at (1,0){$\Lambda$};
  \node[c](B) at (0,0){$X$};
  \node[c](C) at (2,0){$Y$};
  \node[c](D) at (1.5,1){$B$}
  edge[e] (A)
  edge[e] (C);
  \node[c](E) at (0.5,1){$A$}
  edge[e] (A)
  edge[e] (B);
\end{tikzpicture}
\end{center}
\caption{The Bell scenario depicted as a DAG, with hidden variable $\Lambda$.}\label{belldag}
\end{figure}
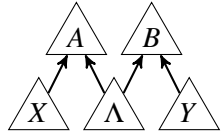
where the $A$ and $B$ nodes represent experiment outcomes in the two wings of the experiment, $X$ and $Y$ are the respective settings, and $\Lambda$ is the common cause (i.e.~the `hidden variable'). 
Writing down such a DAG incorporates various causal assumptions, for example: (i) that the settings are `free', e.g.~there are no edges $\Lambda \to X$ or $\Lambda \to Y$;
 and (ii) that the two wings are causally disconnected from each other (which could arise from spacelike separation between Alice and Bob), e.g.~there is no edge $X\to B$.

Let us now consider random variables associated to the nodes of the DAG.  Only certain probability distributions will be consistent with the causal structure, if it is to have the intended meaning.  As in other treatments, $X\s{1}$ will denote a random variable, while $x\s{1}$ denotes the value of the random variable, and the same label will also be used for the node in the graph associated to this variable (it will be clear from the context which is meant). Sometimes capital letters will also be used to signify sets of random variables, and the lowercase letter a value for each of these variables.  The basic objects of interest will be probability distributions over all the nodes, $P(g)$.  It is convenient to extend this notation to the parents in the following way:
\begin{itemize}
  \item $\PA{X\s{i}}$ is the set of random variables associated to the parents of the node $X\s{i}$;
  \item $\pa{x\s{i}}$ denotes a values of the random variables $\PA{X\s{i}}$.
\end{itemize}

The notion of causality that we now apply has several equivalent forms \cite{pearl}.
Perhaps the most intuitive is that given a random variable $X$, once direct causal influence of the parents has been taken into account by conditioning, then $X$ should be independent of every other node, except for its descendants. For our purposes the following form is the most suggestive:

\begin{definition}[Markov condition]\label{def:markov}
Let $G$ be a DAG.
A probability distribution $P$ is \emph{Markov relative to $G$} if $P$ satisfies
\[
P(x^{(1)},\dots,x^{(n)}) = \prod_i P(x^{(i)}|\pa{x\s{i}}).
\]
\end{definition}

A simple example is given by a probability distribution $P$ that is Markov with respect to the chain $X\to Y\to Z$: this means that $Y$ `screens off' the influence of $X$ from $Z$, i.e.~$P(z|x,y)=P(z|y)$.

\begin{definition}
A \emph{(classical) Bayesian network} is a pair $(P, G)$ where $G$ is a DAG, and $P$ is a probability distribution that is Markov relative to $G$.
\end{definition}

Often, only a subset of the nodes in a Bayesian network represent observable outcomes. 
These are called \emph{observed} nodes, whereas the other nodes are referred to as \emph{latent} or \emph{hidden} nodes.
Latent nodes are usually added by hypothesis in an attempt to explain observed correlations.


We can describe the Bell's theorem in this language \cite{wood}.
If $P$ is Markov relative to the DAG in \cref{belldag} then
\[
P(a,b,x,y,\lambda)=P(a|x,\lambda)P(b|y,\lambda)P(x)P(y)P(\lambda).
\]
After marginalising over $\lambda$, and dividing through by $P(x)P(y)$, we obtain Bell's locality condition, i.e.~\cref{eq:Bell_locality}.
Hence we see that: (i) the idea of a hidden variable $\lambda$ is identical to the existence of a latent node; (ii) Bell's locality condition follows from the Markov condition for the Bell DAG.
In this way, we can see that Bell applied the same basic account of causality as used in Bayesian networks, albeit applied to a particularly simple and intuitive case.  
For more complex DAGs, 
the more general framework is needed.

\subsection{A graphical criterion for independence: $d$-separation}

A Bayesian network specifies a graph and a probability distribution that decomposes `locally' along the edges of the graph.
This means that it encodes certain conditional independences.
But in general, further independences will be derivable from those given directly by the fact that $P$ is Markov with respect to $G$.  
For example, in the Bell DAG, the Markov condition immediately implies that $P(a|x,\lambda,y)=P(a|x,\lambda)$ (sometimes called `parameter-independence' \cite{Shimony2013}).
But the probability calculus also implies that we can marginalise over $\lambda$ to obtain $P(a|x,y)=P(a|x)$, i.e. the no-signalling condition.
In the theory of Bayesian networks, these additional conditional independences are of paramount importance. Clearly they follow from the structure of the graph alone, but deriving them using probability theory can be impractical, especially in more complicated DAGs.
The condition of $d$-separation, developed by Geiger \cite{Geiger1987} and Verma and Pearl \cite{Verma1988}, provides a way to `read off' these conditional independences from the structure of the graph.  

To gain an intuitive understanding of the $d$-separation condition, let us consider the connected Bayesian networks that have three nodes, $X$, $Y$ and $Z$, and two edges.
There are three such networks:
\begin{enumerate}[(i)]
	\item The \emph{chain} $X \rightarrow Z \rightarrow Y$;
	\item The \emph{fork} $X \leftarrow Z \rightarrow Y$; and
	\item The \emph{collider} $X \rightarrow Z \leftarrow Y$.
\end{enumerate}
We can consider whether $P(x,y|z)=P(x|z)P(y|z)$ holds in each of these cases, denoted $\ci{X}{Y}{Z}$.
For the chain and fork, it is immediate that $X$ and $Y$ are conditionally independent given $Z$ in any Markov probability distribution (but need not satisfy marginal independence $p(x,y)=p(x)p(y)$).  
However, in the collider we may not have $\ci{X}{Y}{Z}$, even though $X$ and $Y$ are now marginally independent.
For example, $Z$ could hold the value $1$ when $x=y$, and $0$ otherwise.  
The same prevention of conditional independence may be caused by conditioning on any node in the mutual future of $X$ and $Y$ in a more general DAG.  
Roughly speaking, these observations show that, for sets of nodes $X$, $Y$ and $Z$, conditional independences $\ci{X}{Y}{Z}$ will follow when $Z$ \emph{contains} the middle node of chains and forks, but \emph{excludes} the middle node of colliders.

We shall use the form of $d$-separation originally developed by Lauritzen \textit{et al.}\ \cite{lauritzen}.
Let $G$ be a DAG with disjoint subsets $X$, $Y$ and $Z$. 
Then we define the set $W:=G \setminus \ian{X \cup Y \cup Z}$.  In words, the set $W$ is every node in $G$ that is not in the inclusive past of any node in $X$, $Y$ or $Z$.
Now define a \textit{pseudo-path} from node $P\s{1}$ to node $P\s{p}$ to be a sequence of nodes $(P\s{1},P\s{2},\dots,P\s{p})$ such that, for all $i \in \{1,...,p\}$, $P\s{i} \not \in W$, and $m(P\s{i}) \cap m(P\s{i+1}) \not\subseteq W$.  
That is, a pseudo-path does not intersect $W$, and two sequential elements in a pseudo-path must be adjacent or share a common child that is not in $W$.

\begin{definition}
\label{d:dsep}
Let $G$ be a DAG $G$ with disjoint subsets $X$, $Y$ and $Z$. 
We say that $X$ and $Y$ are \emph{$d$-separated} by $Z$, written $\dsep{X}{Y}{Z}$, if, for all nodes $A \in X$ and $B \in Y$, all pseudo-paths from $A$ to $B$ are non-trivially intersected by $Z$.
\end{definition}

\begin{example}[$d$-separation]
As we would expect, for the chain and fork we have $\dsep{X}{Y}{Z}$, but this fails for the collider. 
Consider the the dotted line in Figure~\ref{fig:collider2}. 
This is a pseudo-path, since $W$ is the empty set in this DAG, and the path has only two sequential elements, with $Z$ as the common child.
However this pseudo-path does not intersect $Z$, and hence $\dsep{X}{Y}{Z}$ fails to hold.
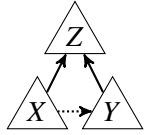
\begin{figure}[h]
	\begin{center}
	\begin{tikzpicture}
\node[c](X) at (0,0) {$X$};
\node[c](Y) at (1,0) {$Y$}
edge[e,densely dotted] (X);
\node[c](Z) at (0.5,1) {$Z$}
edge[e] (X)
edge[e] (Y);
\end{tikzpicture}
	\end{center}
	\caption{A pseudo-path for the collider.}\label{fig:collider2}
\end{figure}
\end{example}

The following theorem establishes the link between the $d$-separation condition and conditional independence.

\begin{thm}[Verma and Pearl \cite{Verma1988}, Meek \cite{Meek95}]\label{t:dsep}
Let $G$ be a DAG with disjoint subsets $X$,$Y$ and $Z$. Then:
\begin{itemize}
\item[(i)] If $P$ is Markov with respect to $G$, then $\dsep{X}{Y}{Z} \Rightarrow \ci{X}{Y}{Z}$.
\item[(ii)] If $\ci{X}{Y}{Z}$ holds for all $P$ which are Markov with respect to $G$, then $\dsep{X}{Y}{Z}$.
\end{itemize}
\end{thm}

Item (i) says that $d$-separation is a sound criterion for conditional independence, and item (ii) says that $d$-separation is complete, \textit{i.e.}~\emph{all} robust conditional independences arise through applying the $d$-separation condition to the underlying DAG. Theorem \Cref{t:dsep} is of central importance to classical Bayesian networks. 
For example, many algorithms for causal inference rely exclusively on conditional independences \cite{pearl}.  

\begin{example}[Conditional independences in the Bell scenario]
Consider again the Bell DAG \cref{belldag}.
We can use the $d$-separation theorem to derive the usual conditional independences, i.e.~the no-signalling conditions.
For example, we have $\dsep{A}{Y}{X}$, which implies $P(a|x,y)=P(a|x)$.
We obtain $\dsep{A}{Y}{X}$ as follows.
We have $W=\{B\}$.
But consider the sequences of nodes between $A$ and $Y$ (the candidate pseudo-paths).
For example, $p_1:=(Y,B,\Lambda,A)$ and $p_2:=(Y,\Lambda,A)$ are two such sequences.
But $p_1$ intersects $W$, and $p_2$ has a pair sequential elements $(\Lambda,A)$ that share a common child in $W$.
Similarly, all sequences of nodes between $A$ and $Y$ fail to be pseudo-paths, and hence
$\dsep{A}{Y}{\emptyset}$.\footnote{
	Note that the Bell DAG here encodes the assumption that the inputs are uncorrelated, i.e.~$P(x,y)=P(x)P(y)$. 
	Hence we obtain a `no-signalling' condition that is stronger than the usual one considered for nonlocality in the Bell setup.
	That is, we obtain $P(a|y)=P(a)$ as well as $P(a|x,y)=P(a|x)$.
	To allow for the possibility that the inputs are correlated, we would use a different DAG, with extra edges $U\to X$ and $U\to Y$, where $U$ represents a correlating variable.
	With this DAG, we obtain only $P(a|x,y)=P(a|x)$, without $P(a|y)=P(a)$, as expected.
}
\end{example}

\section{Generalised Bayesian networks}\label{sec:GBN}

We will now extend the definitions of Bayesian networks to go beyond classical theories.
This will serve as a framework within which to discuss the differences between the allowed set of probability distributions in classical and quantum systems, and even more general cases.  
To do this we shall build on the circuit framework for general probabilistic theories that was developed by Chiribella, D'Ariano and Perinotti (CDP) \cite{Chiribella2010}. 
This provides a graphical approach which is useful when considering DAGs, and their framework imposes very minimal requirements on the theories it encompasses. We describe this CDP framework in Section~\ref{sec:CDP}.  We then introduce our definition of generalised Bayesian networks in Section~\ref{sec:def_GBN}, after which, in Section~\ref{sec:gen_dsep}, we prove that the $d$-separation criterion can be extended to generalised Bayesian networks.

\subsection{The Chiribella-D'Ariano-Perinotti framework}\label{sec:CDP}


The CDP framework provides an abstract description of `circuits' consisting of operations (which include preparations, transformations and observations) connected by propagating systems.  These will be used to describe sources of general correlations in our generalied Bayesian networks.  First, the way in which elements of the circuits compose will be specified (the `operational' part); then the way in which probabilities are attached to circuits will be described.  Together these parts constitute what CDP call an \emph{operational-probabilistic theory}.

\subsubsection{The operational part}

To specify the operational part, we consider a collection of named systems $A,B,C\dots$, including a \em trivial \em system $I$. 
Systems are the inputs and outputs of  \em tests \em $\{\mathcal{C}_i\}_{i\in X}$, which represent a single use of some physical device, e.g.~a Stern-Gerlach device.
 To prevent the input of a test being its own output,  the input and output systems of a test must be distinct, except when both are trivial.
The elements of tests, $\mathcal{C}_i$, represent operationally distinguishable outcomes of the test. They are referred to as \em events\em, and are indexed by a finite number of outcomes $i\in X$. 

For example, for the test corresponding to the use of a Stern-Gerlach device with a spin-half particle, the outcome set would have two elements, corresponding to the two different spin outcomes.
If a test $\{\mathcal{C}_i\}_{i\in X}$ is a singleton, i.e.~if there is only one outcome $i=i_0$, then we say that this is a \emph{deterministic test}.

Below we will find it useful to explicitly include the input and output systems in our notation, so an event with input system $A$ and output $B$ will be represented as $\mathcal{C}_{i \, A}^{B}$. The trivial system will not be included explicitly.\footnote{This mimics the use of tensorial notation by Hardy \cite{hardytensor,hardytensor2}.}
CDP use a graphical notation that builds upon that of Abramsky and Coecke \cite{Abramsky2004a}.
A test $\{\mathcal{C}_i\}$ with input system $A$ and output system $B$ is depicted as:
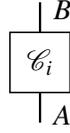
\[
%
\ifx\JPicScale\undefined\def\JPicScale{1}\fi
\psset{unit=\JPicScale mm}
\psset{linewidth=0.3,dotsep=1,hatchwidth=0.3,hatchsep=1.5,shadowsize=1,dimen=middle}
\psset{dotsize=0.7 2.5,dotscale=1 1,fillcolor=black}
\psset{arrowsize=1 2,arrowlength=1,arrowinset=0.25,tbarsize=0.7 5,bracketlength=0.15,rbracketlength=0.15}
\begin{pspicture}(0,0)(12,17)
\pspolygon[linewidth=0.2,fillcolor=white,fillstyle=solid](4,13)
(4,5)
(12,5)
(12,13)(4,13)
\rput(8,9){$\mathcal{C}_i$}
\rput(7,11){}
\psline(8,13)(8,17)
\psline(8,1)(8,5)
\rput(11,2){$A$}
\rput(11,16){$B$}
\end{pspicture}

\]
If the input of a test is the trivial system then it is depicted as
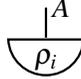
\[
%
\ifx\JPicScale\undefined\def\JPicScale{1}\fi
\psset{unit=\JPicScale mm}
\psset{linewidth=0.3,dotsep=1,hatchwidth=0.3,hatchsep=1.5,shadowsize=1,dimen=middle}
\psset{dotsize=0.7 2.5,dotscale=1 1,fillcolor=black}
\psset{arrowsize=1 2,arrowlength=1,arrowinset=0.25,tbarsize=0.7 5,bracketlength=0.15,rbracketlength=0.15}
\begin{pspicture}(0,0)(10.01,11.01)
\rput(5,3){$\mathcal{\rho}_i$}
\rput(4,4){}
\psline(5,6)(5,10)
\rput{0}(5,6.01){\pscustom[]{\psellipticarc(0,0)(5.01,5.01){-177.06}{-2.94}\closepath}}
\rput(7,10){$A$}
\end{pspicture}

\]
and referred to as a \em preparation-test\em. 
\em Observation-tests \em are the dual notion, for which the output is the trivial system. 
From now on, we shall omit labelling the systems in the graphical notation.

When the output system of $\{\mathcal{C}_i\}$ is the same as the input system of $\{\mathcal{D}_j\}$, they are composed \em in sequence\em, depicted as
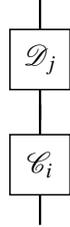
\[
\ifx\JPicScale\undefined\def\JPicScale{1}\fi
\psset{unit=\JPicScale mm}
\psset{linewidth=0.3,dotsep=1,hatchwidth=0.3,hatchsep=1.5,shadowsize=1,dimen=middle}
\psset{dotsize=0.7 2.5,dotscale=1 1,fillcolor=black}
\psset{arrowsize=1 2,arrowlength=1,arrowinset=0.25,tbarsize=0.7 5,bracketlength=0.15,rbracketlength=0.15}
\begin{pspicture}(0,0)(12,31)
\pspolygon[linewidth=0.2,fillcolor=white,fillstyle=solid](4,13)
(4,5)
(12,5)
(12,13)(4,13)
\rput(8,9){$\mathcal{C}_i$}
\rput(7,11){}
\psline(8,13)(8,17)
\psline(8,1)(8,5)
\pspolygon[linewidth=0.2,fillcolor=white,fillstyle=solid](4,27)
(4,19)
(12,19)
(12,27)(4,27)
\rput(8,23){$\mathcal{D}_j$}
\rput(7,25){}
\psline(8,27)(8,31)
\psline(8,15)(8,19)
\end{pspicture}

\]
or symbolically as $\mathcal{C}_{i \, A}^{\, B} \mathcal{D}_{j \, B}^{\, C}$. Otherwise they are composed \em in parallel:\em
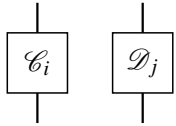
\[
\ifx\JPicScale\undefined\def\JPicScale{1}\fi
\psset{unit=\JPicScale mm}
\psset{linewidth=0.3,dotsep=1,hatchwidth=0.3,hatchsep=1.5,shadowsize=1,dimen=middle}
\psset{dotsize=0.7 2.5,dotscale=1 1,fillcolor=black}
\psset{arrowsize=1 2,arrowlength=1,arrowinset=0.25,tbarsize=0.7 5,bracketlength=0.15,rbracketlength=0.15}
\begin{pspicture}(0,0)(26,17)
\pspolygon[linewidth=0.2,fillcolor=white,fillstyle=solid](4,13)
(4,5)
(12,5)
(12,13)(4,13)
\rput(8,9){$\mathcal{C}_i$}
\rput(7,11){}
\psline(8,13)(8,17)
\psline(8,1)(8,5)
\pspolygon[linewidth=0.2,fillcolor=white,fillstyle=solid](18,13)
(18,5)
(26,5)
(26,13)(18,13)
\rput(22,9){$\mathcal{D}_j$}
\rput(21,11){}
\psline(22,13)(22,17)
\psline(22,1)(22,5)
\end{pspicture}

\] 
or $\mathcal{C}_{i \, A}^{\, B} \mathcal{D}_{j \, C}^{\, D}$.  Each type of composition yields another test, whose outcomes $(i,j)$ are ordered pairs formed by the outcomes $i$ and $j$ of each factor. 
 
If $\mathcal{C}_i$ has input system $A$ and output $B$, and $\mathcal{D}_i$ has input $C$ and output $D$, then their parallel composition has the \em composite systems\em, $AC$ and $BD$, as inputs and outputs respectively.
`Composite system' is a primitive notion for CDP, assumed to satisfy certain basic requirements, and so it is not defined with respect to any other mathematical structure.

\subsubsection{The probabilistic part}
An operational-probabalistic theory is defined as one in which every test from the trivial system to itself (pictorally, a diagram with no input or output wires) is a probability distribution over the outcome set, and where the composition of such tests is given by the corresponding product distribution.

Two tests are called operationally equivalent if substituting one for the other never affects a probability distribution. An operationally equivalent class of observation-events is called an \emph{effect}.

To complete this framework we shall assume the existence of a \emph{unique} deterministic effect $\top_A$ for each system $A$. 
Graphically we denote this as:
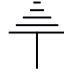
\[
\ifx\JPicScale\undefined\def\JPicScale{1}\fi
\psset{unit=\JPicScale mm}
\psset{linewidth=0.3,dotsep=1,hatchwidth=0.3,hatchsep=1.5,shadowsize=1,dimen=middle}
\psset{dotsize=0.7 2.5,dotscale=1 1,fillcolor=black}
\psset{arrowsize=1 2,arrowlength=1,arrowinset=0.25,tbarsize=0.7 5,bracketlength=0.15,rbracketlength=0.15}
\begin{pspicture}(0,0)(9.38,9.38)
\psline[linewidth=0.25,fillcolor=blue,fillstyle=solid](5.62,0.6)(5.62,5.47)
\psline[linewidth=0.25,fillcolor=blue,fillstyle=solid](1.88,5.39)(9.38,5.39)
\psline[linewidth=0.25,fillcolor=blue,fillstyle=solid](3.75,7.39)(7.5,7.39)
\psline[linewidth=0.25,fillcolor=blue,fillstyle=solid](5.16,9.38)(6.09,9.38)
\psline[linewidth=0.25,fillcolor=blue,fillstyle=solid](2.82,6.38)(8.43,6.38)
\psline[linewidth=0.25,fillcolor=blue,fillstyle=solid](4.69,8.39)(6.55,8.39)
\end{pspicture}

\]
This assumption is referred to by CDP as \emph{causality}. In particular, ignoring the outcome of any observation-test always corresponds to this unique deterministic effect. This assumption is necessary for the comparison to Bayesian networks below to make sense: CDP show that it is equivalent to the assumption that the probability of an outcome at time $t_1$ does not depend on which operation is performed at time $t_2$, where $t_2>t_1$.  Hence the causality assumption can also be thought of as `no-signalling from the future to the past'.

The fact that the deterministic effect is unique trivially impies the following result, which we will use below.

\begin{lemma}
\label{l:compose_effect}
The deterministic effect on a composite system $AB$ is equal to the parallel composition of the deterministic effect on $A$ with the deterministic effect on $B$, or $\top_{AB}=\top_{A}\top_{B}$.
\end{lemma}

We can now give some examples of causal operational-probabilistic theories.
\begin{example}[Quantum theory]
Quantum theory will be our main example of an operational-probabilistic theory.
Systems $A,B,C,\dots$ are associated to complex Hilbert spaces $\Hb_A, \Hb_B, \Hb_C,\dots$; and in particular, the trivial system is given by the one-dimensional space $\Hb_I=\mathbb{C}$. Composite systems are given by the vector space tensor product.

Tests are quantum instruments, i.e. sets of completely positive linear maps that sum to a trace preserving map. In particular, deterministic preparation-tests are unit trace positive operators, and observation-tests are of the form $\Tr(E_i \cdot)$ where $\{E_i\}$ is a POVM. Tests compose in sequence by ordinary composition of maps, and in parallel by the vector tensor product. The unique determinstic effect is $\Tr$.
\end{example}

\begin{example}[Boxworld] 
\emph{Boxworld} \cite{jon} is a theory defined to produce the maximal violation \cite{prbox} of the CHSH inequality. The simplest type of system, called a \emph{gbit}, comes with a pair of two-outcome observation-tests $\{e_1, e_2\}$ and $\{f_1, f_2\}$. For any pair of probabilities $p_e$ and $p_f$ there is exactly one deterministic preparation-test $\omega$ with $e_1(\omega) = p_e$ and $f_1(\omega) = p_f$. Composite systems get the parallel compositions of these, and there is then a unique deterministic preparation-test for any no-signalling distribution on the outcomes. Subject to these requirements, every other mathematically consistent test is included.
\end{example}

\begin{example}[Classical probability theory]\label{ex:classical_CDP}
  We obtain a \emph{classical operational-probabilistic theory} by associating systems $A, B, C, \dots$ with sets $\Lambda_A, \Lambda_B, \Lambda_C$, the trivial system having $\Lambda_I = \{\emptyset\}$. Composite systems are given by the Cartesian product.

  Tests with outcome $i$ from a system $A$ to a system $B$ are given by $p(i, \lambda_B|\lambda_A) \geq 0$, with $\lambda_A$ and $\lambda_B$ ranging over $\Lambda_A$ and $\Lambda_B$ respectively, and $\sum_{i,\lambda_B} p(i,\lambda_B|\lambda_A) = 1$. Tests compose in sequence by multiplying and the summing over the $\lambda$ for the intermediate system, and in parallel by multiplying. The unique deterministic effect is $p(\emptyset|\lambda) = 1$.

A natural question is whether a classical operational-probabilistic theory is, in fact, a Bayesian network.
However, there are two reasons why this is not the case:
\begin{enumerate}
  \item There is no classical conditioning in an operational-probabilistic theory. 
  That is, a test $\{\mathcal{C}_i\}_i$ should be thought of as a device with an output indicating which classical outcome, e.g.~a light that flashes red or green depending on whether spin up or down is detected.
  However, in general a physical device will have `dials', which can be used to control which operation will take place.
  This corresponds to allowing a test $\{\mathcal{C}_i\}_i$ to be a function of a classical input.
  Indeed, this is how the Bell setup is usually conceived, since Alice and Bob each have two possible measurements (observation-tests), and these measurements are chosen based on their input choice, which can be represented as a binary classical random variable. 
  \item 
  An operational-probabilistic theory carries \emph{two} types of information in each circuit element: the systems that `travel' along the wires, and the classical outputs.  The outcome of a test need not tell us everything about the test's output state, even when the relevant system is classical.  Hence a direct interpretation as a DAG, with the outputs of a test translated as the random variables on a node, can easily violate the Markov condition by failing to condition on all the relevant classical information carried by the system.

  For example, consider the following sequence of tests where each system is classical:
  \[
\ifx\JPicScale\undefined\def\JPicScale{1}\fi
\psset{unit=\JPicScale mm}
\psset{linewidth=0.3,dotsep=1,hatchwidth=0.3,hatchsep=1.5,shadowsize=1,dimen=middle}
\psset{dotsize=0.7 2.5,dotscale=1 1,fillcolor=black}
\psset{arrowsize=1 2,arrowlength=1,arrowinset=0.25,tbarsize=0.7 5,bracketlength=0.15,rbracketlength=0.15}
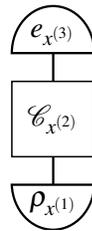
\begin{pspicture}(0,0)(15,27)
\rput(7.21,23.82){}
\rput(9.5,23.24){$e_{x\s{3}}$}
\rput(9.5,0.7){$\rho_{x\s{1}}$}
\rput(8.4,20.74){}
\rput{90}(9.5,3.52){\pscustom[]{\psellipticarc(0,0)(6.27,5.5){92.94}{267.06}\closepath}}
\rput{90}(9.49,20.74){\pscustom[]{\psellipticarc(0,0)(6.26,5.49){-90}{90}\closepath}}
\pspolygon[linewidth=0.2,fillcolor=white,fillstyle=solid](4,16.98)
(4,6.96)
(14.99,6.96)
(14.99,16.98)(4,16.98)
\rput(9.5,11.97){$\mathcal{C}_{x\s{2}}$}
\rput(8.12,14.47){}
\psline(9.5,16.98)(9.5,20.74)
\psline(9.5,3.2)(9.5,6.96)
\end{pspicture}

  \]
  For example, suppose that $\rho$ is the preparation of a coin, which can have either heads or tails facing up, and can be black or white.
  The test $\mathcal{C}$ could change the colour of the coin, but for simplicity let us suppose that each outcome leaves the state of the coin unchanged.
  The classical outputs are as follows: $x\s{1}$ is a bit representing the colour of the coin at $t_1$, $x\s{2}$ is a bit representing the face of the coin at $t_2$, and $x\s{3}$ is a bit representing the colour of the coin at $t_3$. 
  This yields a classical probability distribution $P(x\s{1},x\s{2},x\s{3})$.
  Now, suppose that we try to intepret this circuit as a classical Bayesian network:
  
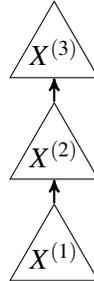
\[
    \begin{tikzpicture}
\node[c,inner sep=-2.4pt](X1) at (0,0) {$X^{(1)}$};
\node[c,inner sep=-2.4pt](X2) at (0,1.35) {$X^{(2)}$}
edge[e] (X1);
\node[c,inner sep=-2.4pt](X3) at (0,2.7) {$X^{(3)}$}
edge[e] (X2);
\end{tikzpicture}

  \]
  The Markov condition implies that $\ci{X^{(3)}}{X^{(1)}}{X^{(2)}}$.
  But if $\rho$ is the preparation of a coin with either side facing up, and in each colour with uniform probability, then $X^{(3)}$ is perfectly correlated with $X^{(1)}$, even conditioning on $X^{(2)}$.  Hence the Markov condition fails to hold.
\end{enumerate}
\end{example}

In the next subsection we shall connect DAGs with generalised probabilistic theories more carefully, overcoming these two problems.

\subsection{Definition of generalised Bayesian networks and examples}\label{sec:def_GBN}
Our aim in this subsection will be to generalise Bayesian networks in a way that can allow non-classical resources.
We begin by splitting the nodes into two types.
\begin{definition}
Let $G$ be a DAG with nodes $V=\{X\s{1},X\s{2},\dots,X\s{m}\}$. 
We shall say that $G$ is a \emph{generalised DAG (GDAG)} if $V$ can be partitioned into two sets of nodes: 
\begin{enumerate}
  \item the \emph{observed nodes} $\{X\s{1},X\s{2},\dots,X\s{n}\}$ (drawn as triangles), and
  \item the \emph{unobserved nodes} $\{X\s{n+1},\dots,X\s{m}\}$ (drawn as circles).
\end{enumerate}
\end{definition}
We choose this terminology because all classical data, \textit{e.g.}~the outcomes of measurements, will be associated to observed nodes.  
On the other hand, the unobserved nodes will replace `latent' random variables with `general resources', e.g.~replacing the source $\lambda$ in the Bell DAG with a general node will allow Alice and Bob to share a quantum state or the state corresponding to a PR box.

We will often apply DAG terminology (parents, children, $d$-separation, etc.)\ to GDAGs. Unless specified otherwise, the relevant definition should simply be applied to the underlying DAG (i.e. ignoring the distinction between observed and unobserved nodes).

We shall assign CDP tests to each node, and hence we shall use the CDP framework.
However, in the previous subsection, we discussed that a circuit element in the CDP framework carries \emph{two} types of data: the classical data associated with an outcome, and the system.
In Example~\ref{ex:classical_CDP} we noted that this makes it problematic to interpret a CDP circuit as a Bayesian network. 
Our framework will address this problem by using generalised DAGs. 
In particular we shall define the \emph{outputs} of observed and unobserved nodes in distinct ways:
\begin{enumerate}
	\item Observed nodes: 
	each observed node will map to a test with no outgoing wires, but will have a classical random variable $X$ assigned to it.
	In the CDP language, an observed node's test has the trivial system as output. Where there is an outgoing edge from an observed node, this means there is a choice of test to be performed at the child node, which depends on the value of the classical variable at the parent. CDP call this a `conditioned test' and show that causality is equivalent to them being well defined.
	\item Unobserved nodes: on the other hand, each unobserved node will output \textit{only} systems, 
	and will not have any non-trivial outcomes assigned to it.
	For convenience of notation we shall associate a classical random variable with every node\footnote{As is the case for classical Bayesian networks, we shall use the same symbol $X\s{i}$ to denote both the node and the random variable associated with the node; context will determine which is being referred to.} $X\s{i}$. However, the random variable associated with unobserved nodes will be trivial, taking only one value with probability one. 
\end{enumerate}
Accordingly, we shall associate a non-trivial probability distribution $P(x\s{1},x\s{2},\dots,x\s{n})$ only with the observed nodes.

More formally, we have:
\begin{definition}
Let $G$ be a generalised DAG. 
Call an edge of $G$ \emph{observed} if it begins on an observed node, and \emph{unobserved} if it begins on an unobserved node.
\end{definition}
\begin{definition}\label{def:GMC}
Let $G$ be a generalised DAG with $m$ nodes, of which the first $n$ are observed.
A probability distribution $P$ over the observed nodes is \emph{generalised Markov} with respect to $G$ if there exists:
\begin{enumerate}
\item a causal operational-probabilistic theory;
\item for every unobserved edge, a distinct system in the theory; and
\item for every node $X\s{i}$, and every value $\cpa{x\s{i}}$ of its observed parents, a test $\cT_{x\s{i}}(\cpa{x\s{i}})_{\GPA{ X\s{i}}}^{\GCH{ X\s{i}}}$ from the composite system $\GPA{ X\s{i}}$ formed by the systems on $X\s{i}$'s incoming unobserved edges to the composite system $\GCH{ X\s{i} }$ formed by the systems on its outgoing unobserved edges, with
\begin{enumerate}
\item an outcome set matching $X\s{i}$ in the case of an observed node, but
\item a 1-element outcome set in the case of an unobserved node
\end{enumerate}
\end{enumerate}
such that
\[
P(x\s{1},x\s{2},\dots,x\s{n}) =
  \prod^{m}_{i=1} \cT_{x\s{i}}(\cpa{ x\s{i} })_{\GPA{ X\s{i}}}^{\GCH{ X\s{i}}}.
\footnote{Although it is not required for our results, it would be nice if $P$ was independent of how the GDAG is described, in particular of how the incoming and outgoing edges of a node are ordered. See \cite{tobias2} for a sketch proof that should carry over to our setting.}
\]
\end{definition}

We say that the \emph{generalised Markov condition (GMC)} is satisfied by a probability distribution $P$ if it is generalised Markov with respect to a given GDAG $G$.
\begin{example}[Prepare and measure]
A randomly chosen preparation followed by a fixed measurement can be depicted as:
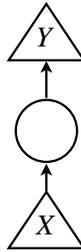
\[
\ifx\JPicScale\undefined\def\JPicScale{1}\fi
\psset{unit=\JPicScale mm}
\psset{linewidth=0.3,dotsep=1,hatchwidth=0.3,hatchsep=1.5,shadowsize=1,dimen=middle}
\psset{dotsize=0.7 2.5,dotscale=1 1,fillcolor=black}
\psset{arrowsize=1 2,arrowlength=1,arrowinset=0.25,tbarsize=0.7 5,bracketlength=0.15,rbracketlength=0.15}
\begin{pspicture}(0,0)(10.62,34.38)
\pscustom[]{\psline(5.62,34.38)(0.62,26.88)
\psline(0.62,26.88)(10.62,26.88)
\psbezier(10.62,26.88)(10.62,26.88)(10.62,26.88)
\psline(10.62,26.88)(5.62,34.38)
\closepath}
\psline{<-}(5.62,26.8)(5.62,21.88)
\psline{<-}(5.62,13.12)(5.62,9.5)
\rput(5.62,5){$X$}
\rput(5.62,30){$Y$}
\rput{90}(5.62,17.5){\psellipse[](0,0)(4.12,-4.07)}
\pscustom[]{\psline(5.63,9.37)(0.62,1.88)
\psline(0.62,1.88)(10.62,1.88)
\psbezier(10.62,1.88)(10.62,1.88)(10.62,1.88)
\psline(10.62,1.88)(5.63,9.37)
\closepath}
\end{pspicture}

\]
The first node, $X$, has no incoming edges. Since it is observed, the corresponding test also has no outgoing systems. Hence it corresponds to a test from the trivial system to itself, i.e. a probability distribution $p_x$. The unobserved node has an incoming edge from $X$ and hence the corresponding test will depend on $x$. It has one outgoing edge and so the test has a single outgoing system, i.e. it is a preparation-test $\rho_x$ for a single system. Finally, the last node corresponds to a test that receives the system from $\rho_x$ and has no outgoing systems, i.e. it is an observation-test $\{e_y\}$. Overall we have:
\[
\ifx\JPicScale\undefined\def\JPicScale{1}\fi
\psset{unit=\JPicScale mm}
\psset{linewidth=0.3,dotsep=1,hatchwidth=0.3,hatchsep=1.5,shadowsize=1,dimen=middle}
\psset{dotsize=0.7 2.5,dotscale=1 1,fillcolor=black}
\psset{arrowsize=1 2,arrowlength=1,arrowinset=0.25,tbarsize=0.7 5,bracketlength=0.15,rbracketlength=0.15}
\begin{pspicture}(0,0)(40,21)
\rput(40,8){}
\rput(20.87,0){$p_x$}
\rput(4.98,12.35){$P(x,y)=$}
\rput(20.88,17.29){$e_y$}
\rput(19.67,18.31){}
\psline(20.88,13.59)(20.89,11.09)
\rput(21,8){$\mathcal{\rho}(x)$}
\rput(19.66,8.65){}
\psline(20.88,11.12)(20.88,14.82)
\rput{90}(20.88,11.13){\pscustom[]{\psellipticarc(0,0)(6.18,6.12){92.94}{267.06}\closepath}}
\rput{90}(20.88,14.82){\pscustom[]{\psellipticarc(0,0)(6.18,6.11){-90}{90}\closepath}}
\end{pspicture}

\]
To interpret this diagram it is useful to recall that the composition of two tests from the trivial system to itself is simply multiplication of probability distributions. Hence $P(x,y) = P(y|x)p_x$ where
\[
\ifx\JPicScale\undefined\def\JPicScale{1}\fi
\psset{unit=\JPicScale mm}
\psset{linewidth=0.3,dotsep=1,hatchwidth=0.3,hatchsep=1.5,shadowsize=1,dimen=middle}
\psset{dotsize=0.7 2.5,dotscale=1 1,fillcolor=black}
\psset{arrowsize=1 2,arrowlength=1,arrowinset=0.25,tbarsize=0.7 5,bracketlength=0.15,rbracketlength=0.15}
\begin{pspicture}(0,0)(40,18)
\rput(40,8){}
\rput(4.98,9.38){$P(y|x)=$}
\rput(20.88,14.31){$e_y$}
\rput(19.67,15.32){}
\psline(20.88,10.62)(20.89,8.13)
\rput(21,5){$\mathcal{\rho}(x)$}
\rput(19.66,5.69){}
\psline(20.88,8.15)(20.88,11.85)
\rput{90}(20.88,8.16){\pscustom[]{\psellipticarc(0,0)(6.16,6.12){92.94}{267.06}\closepath}}
\rput{90}(20.88,11.85){\pscustom[]{\psellipticarc(0,0)(6.15,6.11){-90}{90}\closepath}}
\end{pspicture}

\]
\end{example}

\begin{definition}\label{def:GBN}
A \emph{generalised Bayesian network} is a pair $(P,G)$, such that $G$ is a generalised DAG, and $P$ is generalised Markov with respect to $G$.
\end{definition}

The definition of a generalised Bayesian network is therefore exactly analogous to that of a classical Bayesian network.

\begin{example}[Bell setup]
We can define a generalised Bayesian network corresponding to the Bell scenario as follows:
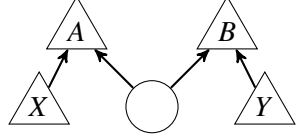
\[
\begin{tikzpicture}
  \node[q](A) at (1.5,0){};
  \node[c](B) at (0,0){$X$};
  \node[c](C) at (3,0){$Y$};
  \node[c](D) at (2.5,1){$B$}
  edge[e] (A)
  edge[e] (C);
  \node[c](E) at (0.5,1){$A$}
  edge[e] (A)
  edge[e] (B);
\end{tikzpicture}

\]
A probability distribution $P$ that is Markov for this generalised DAG is given by:
\[
\ifx\JPicScale\undefined\def\JPicScale{1}\fi
\psset{unit=\JPicScale mm}
\psset{linewidth=0.3,dotsep=1,hatchwidth=0.3,hatchsep=1.5,shadowsize=1,dimen=middle}
\psset{dotsize=0.7 2.5,dotscale=1 1,fillcolor=black}
\psset{arrowsize=1 2,arrowlength=1,arrowinset=0.25,tbarsize=0.7 5,bracketlength=0.15,rbracketlength=0.15}
\begin{pspicture}(0,0)(42.5,15)
\rput(40,8.4){}
\rput(-4,8){$P(a,b,x,y)=$}
\rput(34.25,11.7){$f_b(y)$}
\rput(33.26,12.58){}
\rput(28.12,3.33){$\mathcal{\rho}$}
\rput(33.25,4.26){}
\psline(34.37,5.99)(34.37,9.31)
\rput{0}(27.81,14.14){\pscustom[]{\psellipticarc(0,0)(14.15,14.15){-144.81}{-35.19}\closepath}}
\rput{0}(34.25,9.88){\pscustom[]{\psellipticarc(0,0)(5.01,5.01){-3.44}{183.44}\closepath}}
\rput(21.25,11.7){$e_a(x)$}
\rput(20.26,12.58){}
\psline(21.25,8.51)(21.26,6.36)
\rput(20.25,4.26){}
\psline(21.25,5.99)(21.25,9.58)
\rput{0}(21.25,9.88){\pscustom[]{\psellipticarc(0,0)(5.01,5.01){-3.44}{183.44}\closepath}}
\rput(12,4){$p_x$}
\rput(44,4){$p_y$}
\end{pspicture}

\]
In the special case that the operational theory under consideration is quantum theory, this gives
\[ P(a,b,x,y) = \Tr\left( (E_a(x) \otimes E_b(y)) \rho\right) p_x p_y, \]
where $\rho$ is a bipartite state and $\{E_a(x)\}$ and $\{E_b(y)\}$ are POVMs for each $x,y$. This is indeed the standard quantum model of a Bell experiment.
This example also illustrates that our formalism describes the classical control of tests as a parameterised family of CDP circuits.
\end{example}

A generalised Bayesian network will allow us to explore the consequences of using non-classical resources in place of classical latent variables.
However, we recover classical Bayesian networks if we do not include any unobserved nodes.

\begin{prop}
If all nodes are observed, then a generalised Bayesian network is a classical Bayesian network.
\end{prop}
\begin{proof}
If all nodes are observed, then for every node $X$, we have $\GPA{X}=\GCH{X}=I$. 
That is, the incoming and outoing systems of every test are trivial. 
Then Definition \ref{def:GBN} becomes:
\[
P(x\s{1},x\s{2},\dots,x\s{n}) =
  \prod^{m}_{i=1} \cT_{x\s{i}}(\cpa{ x\s{i} })
\]
For every value of the parents $\pa{x\s{i}} = \cpa{x\s{i}}$, the event $\cT_{x\s{i}}(\cpa{ x\s{i} })$ is a test from $I$ to $I$ and hence a probability distribution on $x\s{i}$.  We can then define
\[
P(x^{(i)}|\pa{x\s{i}}) := \cT_{x\s{i}}(\cpa{ x\s{i} }),
\]
giving a set of conditional probabilities, which, since the composition of tests from $I$ to $I$ is just multiplication, satisfies Definition \ref{def:markov}.
\end{proof}



For a given GDAG $G$, we can identify the following sets of probabilities that are generalised Markov with respect to $G$:
\begin{enumerate}
	\item The set $\Gs$ of probabilities that are generalised Markov for any operational theory.
	\item The set $\Qs$ of probabilities that are generalised Markov for quantum theory.
	\item The set $\Cs$ of probabilities that are generalised Markov for classical probability theory.
\end{enumerate}


Since classical probability theory can be embedded into quantum theory by using diagonal operators, we have $\Cs \subseteq \Qs \subseteq \Gs$ for all GDAGs.

$\Cs$ is closely related to the standard Markov condition on DAGs, with our distinction between observed and unobserved nodes becoming the distinction between observed and latent variables. This is a second sense in which the GMC generalises the usual Markov condition:


\begin{lemma}
\label{l:classicalGMNs}
Let $(P,G)$ be a generalised Bayesian network with $P\in\Cs$.
Then there exists a classical Bayesian network $(P',G')$ where $G'$ is the underlying DAG for $G$, and $P$ and $P'$ agree on the observed nodes defined by $G$.
\end{lemma}

\begin{proof} 
From \cref{def:GBN} and \cref{ex:classical_CDP}, if a generalised Bayesian network $(P,G)$ has $P\in\Cs$, then each node $X\s{i}$ has associated to it a probability distribution $p(x\s{i}, \lambda_{\GCH{ X\s{i}}}|\lambda_{\GPA{ X\s{i}}}, \cpa{ x\s{i} } )$, where $x\s{i}$ is the output, $\lambda_{\GPA{ X\s{i}}}$ is the classical state associated to the incoming edges from unobserved nodes, $\lambda_{\GCH{ X\s{i}}}$ is the classical state associated to the outgoing edges if the node is unobserved (and is trivial otherwise), and $\cpa{ x\s{i} }$ is the output of the observed parents.
In this case, we can define a classical random variable $Y\s{i}$, with values referred to as $y\s{i}$, for each node: for observed nodes this is simply the output random variable, so that $y \s{i}:=x \s{i}$, whereas for unobserved nodes it ranges over the classical states on the set of all the outgoing edges, so that $y \s{i}:= \lambda_{\GCH{ X\s{i}}}$.  
We can now define a probability distribution $P'(y \s{i} | \pa y \s{i})$ from $p(x\s{i}, \lambda_{\GCH{ X\s{i}}}|\lambda_{\GPA{ X\s{i}}}, \cpa{ x\s{i} } )$ in the obvious way. 
This implies that $P'$ is Markov with respect to the underlying DAG of $G$. 
Hence 
we obtain a classical Bayesian network $(P',G')$, and $P'$ agrees with $P$ on the observed nodes of $G$ by construction.
 \end{proof}

Since classical operational-probabilistic theory is defined using a canonical observation-test, and we only consider tests with a finite number of outcomes, $\Cs$ corresponds to classical probability distributions where all variables, including latent ones, are finite. The results of \cite{tobias2} would suggest that this gives observable probability distributions that are dense in the set that includes infinite-valued latent variables. However, it is very much an open question whether or not these sets are in fact equal, although this is known to be the case in the Bell scenario \cite{finethm}.

Finally, we will use $\Is$ to denote the set of probabilities that satisfy all of the observable conditional independences that follow from $d$-separation. In this notation, the first part of \cref{t:dsep} (along with lemma \ref{l:classicalGMNs}) gives $\Cs \subseteq \Is$ for all GDAGs. We will now strengthen this to $\Gs \subseteq \Is$.

\subsection{Extending $d$-separation to generalised Bayesian networks}\label{sec:gen_dsep}

In generalised probabilistic theories, no-signalling is still valid.  Therefore it is to be expected that a generalisation of theorem \ref{t:dsep}, when applied to three disjoint subsets of observed nodes, should obtain.  However, the standard proofs of the soundness part of Theorem \ref{t:dsep} [i.e., item (i)] make use of conditioning on latent variables, the analogue of which is unclear in the general case.\footnote{See \cite{neutral} for progress towards such a concept in the case of quantum theory.}
However, by reformulating $d$-separation before proving the generalisation, an alternative proof can be found that does not rely on conditioning on latent variables, and as a result can be more easily generalised.

\begin{lemma}[Proof in \cref{proofsec}]\label{l:dsep}
Let $G$ be a DAG with disjoint subsets $X$, $Y$ and $Z$, and let $W=G \backslash \ian{X \cup Y \cup Z}$.
Then $X$ and $Y$ are $d$-separated by $Z$ if and only if there exist sets of nodes $U$ and $V$ such that $\{U,V,Z,W\}$ is a partition of G, and
\begin{gather}
\label{e:dsep1}
X \subseteq U, \,Y \subseteq V, \\
\label{e:dsep3}
m(U) \cap m(V) \subseteq W.
\end{gather}
\end{lemma}

\subsection{An example}

We seek a generalisation of theorem \ref{t:dsep} from classical to generalised Bayesian networks.
The following example is intended to clarify why this is reasonable, and also to elucidate the proof.

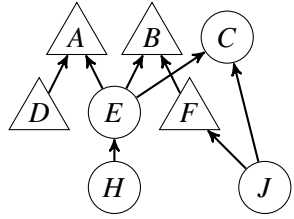
\begin{figure}
\begin{center}
\begin{tikzpicture}
  \node[q](X5) at (1,0) {$H$};
  \node[q](X7) at (1,1) {$E$}
  edge[e] (X5);
  \node[q](X8) at (3,0) {$J$};
  \node[c](X4) at (0,1) {$D$};
  \node[c](X6) at (2,1) {$F$}
  edge[e] (X8);
  \node[c](X2) at (0.5,2) {$A$}
  edge[e] (X4)
  edge[e] (X7);
  \node[c](X3) at (1.5,2) {$B$}
  edge[e] (X6)
  edge[e] (X7);
 \node[q](X1) at (2.5,2) {$C$}
  edge[e] (X7)
  edge[e] (X8);
\end{tikzpicture}
\end{center}
\caption{An example for $d$-separation}
\label{f:dsepexmaple}
\end{figure}

Consider the GDAG depicted in \cref{f:dsepexmaple}. 
This is the Bell GDAG with three extra unobserved nodes, $C$, $J$ and $H$, added.  
Intuitively, the addition of these nodes does nothing to alter the possible GMC probability distributions on the outcomes of the observed nodes. 
For example, the standard no-signalling conditions should still be satisfied.  
To investigate this, let $X:=\{A,D\}$ and
$Y:=\{F\}$, and let $Z$ be empty. 
These sets satisfy the conditions of \cref{l:dsep} with $U:= \{A,D,E,H\}$, $V:=\{F,J\}$ and $W:=\{B,C\}$.  
Hence $X$ and $Y$ are $d$-separated by the empty set.
Therefore, to prove the soundness of the $d$-separation criterion in our setting, we need to show that $X$ and $Y$ are independent in any GMC probability distribution on this graph.

To establish this, we only need to consider $P(x,y)=P(a,d,f)$, which will be the marginal of a probability distribution that satisfies the GMC with respect to the whole GDAG.  
$P(a,d,f)$ can therefore be represented graphically as:
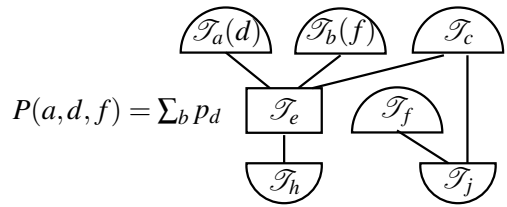
\[
\ifx\JPicScale\undefined\def\JPicScale{1}\fi
\psset{unit=\JPicScale mm}
\psset{linewidth=0.3,dotsep=1,hatchwidth=0.3,hatchsep=1.5,shadowsize=1,dimen=middle}
\psset{dotsize=0.7 2.5,dotscale=1 1,fillcolor=black}
\psset{arrowsize=1 2,arrowlength=1,arrowinset=0.25,tbarsize=0.7 5,bracketlength=0.15,rbracketlength=0.15}
\begin{pspicture}(0,0)(61,28.88)
\rput(40,8){}
\rput(5,15){$P(a,d,f)=$}
\rput(23.13,13.69){}
\psline(31.88,11.88)(31.88,8.12)
\rput(23.12,5.87){}
\rput(24.25,25.5){$\mathcal{T}_a(d)$}
\psline(30,18.12)(24.12,22.88)
\rput{0}(24.12,22.88){\pscustom[]{\psellipticarc(0,0)(6,6){0}{180}\closepath}}
\pspolygon[](26.88,18.12)(36.88,18.12)(36.88,12.19)(26.88,12.19)
\rput(19.38,15){$\sum_b p_d$}
\rput(31.88,15){$\mathcal{T}_e$}
\psline(35.62,18.12)(53.12,22.5)
\rput(39.38,25.62){$\mathcal{T}_b(f)$}
\psline(33.75,18.12)(39.38,22.5)
\rput{0}(39.38,22.75){\pscustom[]{\psellipticarc(0,0)(6,6){0}{180}\closepath}}
\rput(55,25.62){$\mathcal{T}_c$}
\rput{0}(55,22.88){\pscustom[]{\psellipticarc(0,0)(6,6){0}{180}\closepath}}
\rput(46.88,15){$\mathcal{T}_f$}
\rput{0}(46.88,12.25){\pscustom[]{\psellipticarc(0,0)(6,6){0}{180}\closepath}}
\rput(31.88,5.62){$\mathcal{T}_h$}
\rput{0}(31.88,7.82){\pscustom[]{\psellipticarc(0,0)(5.01,5.01){176.54}{363.46}\closepath}}
\rput(55,5.62){$\mathcal{T}_j$}
\rput{90}(55,7.82){\pscustom[]{\psellipticarc(0,0)(5.01,5.01){86.54}{273.46}\closepath}}
\psline(56.25,22.5)(56.25,8.12)
\psline(46.88,12.5)(53.75,8.12)
\end{pspicture}

\]
To be consistent with our motivation, it should not be necessary to mention the nodes in $B$ and $C$ when defining this probability distribution, because they are to the future of all of $A$, $D$ and $F$.  This is indeed the case: this probability distribution still satisfies the GMC with respect to the graph with these two nodes removed.  To see this, note that the outcome $b$ only appears in the effect $\cT_b(f)_{\, E\rightarrow B}$ above, and so summing over all possible outcomes in this factor gives the unique deterministic effect.  The test $\cT_{c \, E\rightarrow C,J \rightarrow C}$ is also a deterministic effect, on $(E\rightarrow C)(J \rightarrow C)$.  We use \cref{l:compose_effect}, which states that the deterministic effect on a product of systems is the product of the deterministic effect on the systems separately.  This gives:
\[
\ifx\JPicScale\undefined\def\JPicScale{1}\fi
\psset{unit=\JPicScale mm}
\psset{linewidth=0.3,dotsep=1,hatchwidth=0.3,hatchsep=1.5,shadowsize=1,dimen=middle}
\psset{dotsize=0.7 2.5,dotscale=1 1,fillcolor=black}
\psset{arrowsize=1 2,arrowlength=1,arrowinset=0.25,tbarsize=0.7 5,bracketlength=0.15,rbracketlength=0.15}
\begin{pspicture}(0,0)(117.5,28.88)
\rput(5,15){$P(a,d,f)=$}
\rput(40.62,7.38){}
\rput(23.75,13.07){}
\psline(32.5,11.25)(32.5,7.5)
\rput(23.74,5.25){}
\rput(24.88,24.88){$\mathcal{T}_a(d)$}
\psline(30.62,17.5)(24.75,22.25)
\rput{0}(24.75,22.25){\pscustom[]{\psellipticarc(0,0)(6,6){0}{180}\closepath}}
\pspolygon[](27.5,17.5)(37.5,17.5)(37.5,11.56)(27.5,11.56)
\rput(32.5,14.38){$\mathcal{T}_e$}
\psline(36.25,17.5)(47.5,21.88)
\psline(34.38,17.5)(40,21.88)
\rput(47.5,14.38){$\mathcal{T}_f$}
\rput{0}(47.5,11.62){\pscustom[]{\psellipticarc(0,0)(6,6){0}{180}\closepath}}
\rput(32.5,5){$\mathcal{T}_h$}
\rput{0}(32.5,7.19){\pscustom[]{\psellipticarc(0,0)(5.01,5.01){176.51}{363.49}\closepath}}
\rput(55.62,5){$\mathcal{T}_j$}
\rput{0}(55.62,7.19){\pscustom[]{\psellipticarc(0,0)(5.01,5.01){176.51}{363.49}\closepath}}
\psline(56.88,21.88)(56.88,7.5)
\psline(47.5,11.88)(54.38,7.5)
\rput(18.75,14.38){$p_d$}
\rput(19.38,15){}
\psline[linewidth=0.25,fillcolor=blue,fillstyle=solid](47.5,21.88)(47.5,25)
\psline[linewidth=0.25,fillcolor=blue,fillstyle=solid](45,24.95)(50,24.95)
\psline[linewidth=0.25,fillcolor=blue,fillstyle=solid](46.25,26.23)(48.75,26.23)
\psline[linewidth=0.25,fillcolor=blue,fillstyle=solid](47.19,27.5)(47.82,27.5)
\psline[linewidth=0.25,fillcolor=blue,fillstyle=solid](45.63,25.58)(49.37,25.58)
\psline[linewidth=0.25,fillcolor=blue,fillstyle=solid](46.88,26.87)(48.13,26.87)
\psline[linewidth=0.25,fillcolor=blue,fillstyle=solid](40,21.92)(40,25.05)
\psline[linewidth=0.25,fillcolor=blue,fillstyle=solid](37.5,25)(42.5,25)
\psline[linewidth=0.25,fillcolor=blue,fillstyle=solid](38.75,26.27)(41.25,26.27)
\psline[linewidth=0.25,fillcolor=blue,fillstyle=solid](39.69,27.55)(40.32,27.55)
\psline[linewidth=0.25,fillcolor=blue,fillstyle=solid](38.13,25.63)(41.87,25.63)
\psline[linewidth=0.25,fillcolor=blue,fillstyle=solid](39.38,26.92)(40.63,26.92)
\psline[linewidth=0.25,fillcolor=blue,fillstyle=solid](56.88,21.88)(56.88,25)
\psline[linewidth=0.25,fillcolor=blue,fillstyle=solid](54.38,24.95)(59.38,24.95)
\psline[linewidth=0.25,fillcolor=blue,fillstyle=solid](55.63,26.23)(58.13,26.23)
\psline[linewidth=0.25,fillcolor=blue,fillstyle=solid](56.57,27.5)(57.19,27.5)
\psline[linewidth=0.25,fillcolor=blue,fillstyle=solid](55.01,25.58)(58.75,25.58)
\psline[linewidth=0.25,fillcolor=blue,fillstyle=solid](56.26,26.87)(57.5,26.87)
\rput(64.38,15.62){$=$}
\rput(95,8){}
\rput(78.13,13.69){}
\psline(86.88,11.88)(86.88,8.12)
\rput(78.12,5.87){}
\rput(79.25,25.5){$\mathcal{T}'_a(d)$}
\psline(85,18.12)(79.12,22.88)
\rput{0}(79.12,22.88){\pscustom[]{\psellipticarc(0,0)(6,6){0}{180}\closepath}}
\pspolygon[](81.88,18.12)(91.88,18.12)(91.88,12.19)(81.88,12.19)
\rput(86.88,15){$\mathcal{T}'_e$}
\rput(101.88,15){$\mathcal{T}'_f$}
\rput{0}(101.88,12.25){\pscustom[]{\psellipticarc(0,0)(6,6){0}{180}\closepath}}
\rput(110,5.62){$\mathcal{T}_j$}
\rput{0}(110,7.82){\pscustom[]{\psellipticarc(0,0)(5.01,5.01){176.51}{363.49}\closepath}}
\psline(101.88,12.5)(108.75,8.12)
\rput(73.12,15){$p_d$}
\rput(73.75,15.62){}
\rput(86.88,5.62){$\mathcal{T}_h$}
\rput{0}(86.88,7.82){\pscustom[]{\psellipticarc(0,0)(5.01,5.01){176.51}{363.49}\closepath}}
\rput(117.5,15){,}
\end{pspicture}

\]
where in the last diagram we define the primed tests as the product of the unprimed tests with any following deterministic effects, for example in the case of $E$,
\begin{equation}
\cT_{e \, H \rightarrow E}^{\prime \, E \rightarrow A}= \cT_{e \, H \rightarrow E}^{E \rightarrow A,E \rightarrow B, E \rightarrow C} \top_{E \rightarrow B}\top_{E \rightarrow C}.
\end{equation}
This result is equivalent to the statement that $P(a,d,f)$ fulfils the GMC for the original GDAG with $B$ and $C$ removed. Once this is done, we only need to note that the circuit has divided into two pieces, one referring to $ad$ but not $f$, and one referring to $f$ but not $ad$.  Recalling that the definition of operational-probabilistic theories requires that tests from $I$ to $I$ compose by multiplication, this establishes that $P(a,d,f)=P(a,d)P(f)$.

There are two main steps in this example, which are both relevant to the general case.  The first was to see that all nodes in $W$ (that is, $B$ and $C$) can be removed from the GDAG, in the following sense: if the probability distribution $P(x,y)$ fulfils the GMC on the whole graph $G$ then its restriction to  $G'=G \backslash W$ fulfils the GMC on $G'$.  Above this is symbolised by absorbing the deterministic effects corresponding to outcomes of nodes in $W$ into the preceding test.  Secondly, after this step, the circuit separates into two parts, and hence the probability distribution can be seen to factorise in the required way.

\subsection{The $d$-separation condition: general case}

We now seek to show that, as in the above example, $d$-separation in a GDAG $G$ implies conditional independence for all probability distributions that are GMC with respect to $G$. 

\begin{lemma}[Proof in \cref{proofsec}]
\label{l:part1}
Consider a GDAG $G$ and a subset $W \subseteq G$ that contains all of its own descendants. If probability distribution $P(g)$ fulfils the GMC on $G$ then the probability distribution $P(g')$ (derived from $P(g)$ by marginalising over outcomes in $W$) fulfils the GMC on $G'= G \backslash W$.
\end{lemma}
This lemma can be applied to eliminate the set $W$ in the reformulation of $d$-separation given above, simplifying our task to proving the following.

\begin{lemma}[Proof in \cref{proofsec}]
\label{l:part2}
Let $X$, $Y$ and $Z$ be disjoint sets of observed nodes in a GDAG $G'$. Suppose $G'$ can be partitioned into $\{U,V,Z\}$ such that
\begin{gather}
\label{e:new1}
X \subseteq U, \,Y \subseteq V \\
\label{e:new2}
m(U) \cap m(V) = \emptyset.
\end{gather}
then $\ci{X}{Y}{Z}$ in any GMC probability distribution on $G'$.
\end{lemma}

Finally, we can prove our $d$-separation theorem.

\begin{thm}
\label{t:dsep_general}

Let $G$ be a generalised DAG with disjoint observed subsets $X$, $Y$ and $Z$. Then
\begin{itemize}
\item[(i)] If $P$ is generalised Markov with respect to $G$, then $\dsep{X}{Y}{Z} \Rightarrow \ci{X}{Y}{Z}$.
\item[(ii)] If $\ci{X}{Y}{Z}$ holds for all $P$ which are generalised Markov with respect to $G$, then $\dsep{X}{Y}{Z}$.
\end{itemize}
\end{thm}

\begin{proof}
To prove item (i), we combine \cref{l:part1} and \cref{l:part2}.  
Item (ii) is a consequence of $\Cs \subseteq \Gs$ and the classical \cref{t:dsep} part (ii).
\end{proof}

\section{Beyond conditional independence: quantitative bounds on correlations}\label{sec:bounds}

In the Bell scenario, Bell inequalities limit the classical correlations (establishing $\Cs \subsetneq \Qs$), and Tsirelson inequalities limit the quantum correlations (establishing $\Qs \subsetneq \Gs$). 
What limits the correlations in a general probabilistic theory? 
In the Bell scenario, a general probabilistic theory is limited \emph{only} by the no-signalling principle (see for example \cite{jon}).  
In our notation, this means that $\Gs = \Is$ for Bell GDAGs. 
Here we show that this fact does not extend to every scenario, i.e.~we provide examples for which $\Gs\subsetneq\Is$.
In other words, causal structure can impose quantitative limits \emph{beyond} the conditional independences between observed nodes, \emph{independently} of the precise physical theory under consideration.

\subsection{The triangle}\label{trisec}

\begin{figure}
  	\begin{center}
	\raisebox{-0.5\height}{
\begin{tikzpicture}[scale=1.2]
  \node[q](A) at (0,0){};
  \node[q](B) at (1,0){};
  \node[q](C) at (2,0){};
  \node[c](D) at (0,1){$A$}
    edge[e] (A)
    edge[e] (B);
  \node[c](E) at (1,1){$B$}
    edge[e] (A)
    edge[e] (C);
  \node[c](F) at (2,1){$C$}
    edge[e] (B)
    edge[e] (C);
  \end{tikzpicture}}
  \hspace{5em}
\raisebox{-0.5\height}{
\begin{tikzpicture}[scale=0.6]
  \node[q](A) at (0,0) {};
  \node[q](B) at (1,1.41) {};
  \node[q](C) at (-1,1.41) {};
  \node[c](D) at (2,0) {$B$}
  edge[e] (A)
  edge[e] (B);
  \node[c](E) at (-2,0) {$A$}
  edge[e] (A)
  edge[e] (C);
  \node[c](E) at (0,2.83) {$C$}
  edge[e] (B)
  edge[e] (C);
\end{tikzpicture}}
	\end{center}
\caption{The `triangle' GDAG drawn in two different ways.}\label{triangledag}
\end{figure}
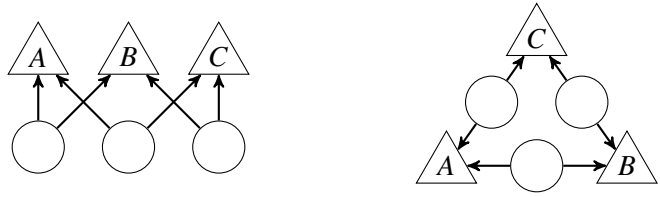

The triangle scenario, shown in \cref{triangledag}, has already received some interest in quantum foundations \cite{bilocal,tobias,chaves} and the causality literature \cite{commonancestors}. 
Branciard et al.~initially introduced the scenario
with definitions matching our $\Cs$ and $\Qs$  \cite{bilocal}. 
It was noted that understanding the classical correlations $\Cs$ in this scenario is much more mathematically challenging than in the Bell scenario. 
Nevertheless, Fritz showed that there exist quantum correlations for this scenario which cannot be reproduced using classical sources, i.e.~$\Cs \subsetneq \Qs$ \cite{tobias}. 
A key part of this proof was showing that any $P \in \Cs$ satisfies a `monogamy' inequality:
\begin{equation}
  \mutual{A}{B} + \mutual{B}{C} \leq \shan{B}. \label{tobiasineq}
\end{equation}
In other words, the stronger the correlations between $A$ and $B$, the weaker must be the correlations between $B$ and $C$.

This has some interesting consequences.
For example, note that there are no independences between observed nodes for this GDAG. 
Hence the `perfectly correlated bits' distribution $P(0,0,0) = P(1,1,1) = \frac12$ is in $\Is$.
However, this perfect correlation violates \cref{tobiasineq}, and hence cannot be produced using classical sources. 

Here we show that \cref{tobiasineq} this holds for any $P \in \Gs$, and hence perfect correlation cannot be produced in this GDAG using any generalised probabilistic theory.
In other words, $\Gs\subsetneq \Is$.
We do this by first proving an important fact about $\Gs$ in this scenario:
\begin{thm}\label{trithm}
  Suppose $P \in \Gs$ for the GDAG in \cref{triangledag}. Then there exists another probability distribution $P'$, such that:
  \begin{enumerate}
    \item\label{tri1} $P'(a,c) = P(a)P(c)$,
    \item\label{tri2} $P'(a,b) = P(a,b)$, and
    \item\label{tri3} $P'(b,c) = P(b,c)$.
  \end{enumerate}
\end{thm}
For a given $P$, the existence of $P'$ is then a linear feasibility problem (studied in \cite{ncycle1,ncycle2}), and hence an efficiently checkable necessary condition for $P \in \Gs$ (and thus also for $\Qs$ and $\Cs$).
\begin{proof}
By the definition of $\Gs$, there exists a causal operational-probabilistic theory with preparations $\rho$, $\sigma$, $\tau$ and observation-tests $\{e_a\}$, $\{f_b\}$, $\{g_c\}$ such that
\[
\ifx\JPicScale\undefined\def\JPicScale{1}\fi
\psset{unit=\JPicScale mm}
\psset{linewidth=0.3,dotsep=1,hatchwidth=0.3,hatchsep=1.5,shadowsize=1,dimen=middle}
\psset{dotsize=0.7 2.5,dotscale=1 1,fillcolor=black}
\psset{arrowsize=1 2,arrowlength=1,arrowinset=0.25,tbarsize=0.7 5,bracketlength=0.15,rbracketlength=0.15}
\begin{pspicture}(0,0)(64.62,16.13)
\rput(5,7.92){$P(a,b,c)=$}
\rput(21.11,7.14){}
\rput(21.1,0.56){}
\rput(21,13.3){$e_a$}
\psline(17.5,4.8)(17.5,11.05)
\rput{0}(20.63,11.05){\pscustom[]{\psellipticarc(0,0)(5,5){-0.57}{180.57}\closepath}}
\rput(21.05,2.43){$\rho$}
\rput{90}(20.75,4.74){\pscustom[]{\psellipticarc(0,0)(4.94,4.94){90.21}{269.79}\closepath}}
\rput(45,0.25){}
\rput(40.11,7.09){}
\rput(40.1,0.51){}
\rput(40,13.25){$e_b$}
\psline(37.5,4.83)(23.75,11.05)
\rput{0}(40,11.05){\pscustom[]{\psellipticarc(0,0)(5,5){-0.57}{180.57}\closepath}}
\rput(40.05,2.38){$\sigma$}
\rput{0}(40.05,4.72){\pscustom[]{\psellipticarc(0,0)(4.94,4.94){-179.67}{-0.33}\closepath}}
\psline(23.75,4.8)(37.5,11)
\rput(59.73,7.12){}
\rput(59.72,0.53){}
\rput(59.62,13.27){$e_c$}
\psline(62.5,4.8)(62.5,11.05)
\rput{0}(59.38,11.13){\pscustom[]{\psellipticarc(0,0)(5,5){-0.57}{180.57}\closepath}}
\rput(59.68,2.4){$\tau$}
\rput{0}(59.68,4.74){\pscustom[]{\psellipticarc(0,0)(4.95,4.95){-179.73}{-0.27}\closepath}}
\psline(42.5,4.8)(56.88,11.05)
\psline(56.88,4.8)(42.5,11.05)
\end{pspicture}

\]
We can use these, along with the unique deterministic effect, to define
\[
\ifx\JPicScale\undefined\def\JPicScale{1}\fi
\psset{unit=\JPicScale mm}
\psset{linewidth=0.3,dotsep=1,hatchwidth=0.3,hatchsep=1.5,shadowsize=1,dimen=middle}
\psset{dotsize=0.7 2.5,dotscale=1 1,fillcolor=black}
\psset{arrowsize=1 2,arrowlength=1,arrowinset=0.25,tbarsize=0.7 5,bracketlength=0.15,rbracketlength=0.15}
\begin{pspicture}(0,0)(64.62,16.33)
\rput(5,8.13){$P'(a,b,c)=$}
\rput(21.11,7.35){}
\rput(21.1,0.76){}
\rput(21,13.5){$e_a$}
\psline(17.5,5)(17.5,11.25)
\rput{90}(20.63,11.2){\pscustom[]{\psellipticarc(0,0)(5.05,5){-90}{90}\closepath}}
\rput(21.05,2.63){$\rho$}
\rput{0}(20.75,4.64){\pscustom[]{\psellipticarc(0,0)(4.95,4.64){176.54}{363.46}\closepath}}
\rput(45,0.45){}
\rput(40.11,7.29){}
\rput(40,13.45){$e_b$}
\psline(30.62,5)(23.75,11.25)
\rput{90}(40,11.2){\pscustom[]{\psellipticarc(0,0)(5.05,5){-90}{90}\closepath}}
\psline(23.75,5)(37.5,11.2)
\rput(59.73,7.32){}
\rput(59.72,0.73){}
\rput(59.62,13.48){$e_c$}
\psline(62.5,5)(62.5,11.25)
\rput{90}(59.37,11.28){\pscustom[]{\psellipticarc(0,0)(5.05,5){-90}{90}\closepath}}
\rput(59.68,2.61){$\tau$}
\rput{0}(59.68,4.64){\pscustom[]{\psellipticarc(0,0)(4.95,4.64){176.54}{363.46}\closepath}}
\psline(50,5)(56.88,11.25)
\psline(56.88,5)(42.5,11.25)
\psline[linewidth=0.25,fillcolor=blue,fillstyle=solid](44.38,5.06)(44.38,7.45)
\psline[linewidth=0.25,fillcolor=blue,fillstyle=solid](42.5,7.42)(46.25,7.42)
\psline[linewidth=0.25,fillcolor=blue,fillstyle=solid](43.44,8.4)(45.31,8.4)
\psline[linewidth=0.25,fillcolor=blue,fillstyle=solid](44.14,9.38)(44.61,9.38)
\psline[linewidth=0.25,fillcolor=blue,fillstyle=solid](42.97,7.9)(45.78,7.9)
\psline[linewidth=0.25,fillcolor=blue,fillstyle=solid](43.91,8.89)(44.84,8.89)
\rput(38.85,0.71){}
\rput(33.12,2.5){$\sigma$}
\rput{0}(33.25,4.64){\pscustom[]{\psellipticarc(0,0)(4.95,4.64){176.54}{363.46}\closepath}}
\rput(53.22,0.71){}
\rput(47.5,2.5){$\sigma$}
\rput{0}(47.63,4.64){\pscustom[]{\psellipticarc(0,0)(4.95,4.64){176.54}{363.46}\closepath}}
\psline[linewidth=0.25,fillcolor=blue,fillstyle=solid](36.25,5.06)(36.25,7.45)
\psline[linewidth=0.25,fillcolor=blue,fillstyle=solid](34.38,7.42)(38.12,7.42)
\psline[linewidth=0.25,fillcolor=blue,fillstyle=solid](35.31,8.4)(37.19,8.4)
\psline[linewidth=0.25,fillcolor=blue,fillstyle=solid](36.02,9.38)(36.48,9.38)
\psline[linewidth=0.25,fillcolor=blue,fillstyle=solid](34.85,7.9)(37.65,7.9)
\psline[linewidth=0.25,fillcolor=blue,fillstyle=solid](35.78,8.89)(36.71,8.89)
\end{pspicture}

\]
Notice that $P' \in \Gs'$ for the GDAG $\Gs'$ depicted in \cref{primeddag}. Since $A$ is $d$-separated from $C$ in this GDAG, \cref{t:dsep_general} gives $P'(a,c)=P'(a)P'(c)$, which, once we have also established \cref{tri2,tri3}, gives \cref{tri1}.

\begin{figure}
\begin{center}
\begin{tikzpicture}
\node[q](A) at (0,0){};
\node[q](B1) at (1,0){};
\node[q](B2) at (3,0){};
\node[q](C) at (4,0){};
\node[c](D) at (0,1){$A$}
edge[e] (A)
edge[e] (B1);
\node[c](E) at (2,1){$B$}
edge[e] (A)
edge[e] (C);
\node[c](F) at (4,1){$C$}
edge[e] (B2)
edge[e] (C);
\end{tikzpicture}
\end{center}
\caption{The GDAG for $P'$.}
\label{primeddag}
\end{figure}
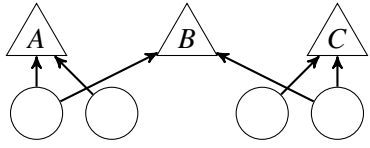

Using \cref{l:compose_effect}, we find
\[
\ifx\JPicScale\undefined\def\JPicScale{1}\fi
\psset{unit=\JPicScale mm}
\psset{linewidth=0.3,dotsep=1,hatchwidth=0.3,hatchsep=1.5,shadowsize=1,dimen=middle}
\psset{dotsize=0.7 2.5,dotscale=1 1,fillcolor=black}
\psset{arrowsize=1 2,arrowlength=1,arrowinset=0.25,tbarsize=0.7 5,bracketlength=0.15,rbracketlength=0.15}
\begin{pspicture}(0,0)(68,16.05)
\rput(-5.62,7.5){$P(a,b)=\sum_c P(a,b,c)=$}
\rput(21.11,7.14){}
\rput(21.1,0.56){}
\rput(21,13.3){$e_a$}
\psline(17.5,4.8)(17.5,11.05)
\rput{0}(20.63,11.05){\pscustom[]{\psellipticarc(0,0)(5,5){-0.57}{180.57}\closepath}}
\rput(21.05,2.43){$\rho$}
\rput{0}(20.75,4.74){\pscustom[]{\psellipticarc(0,0)(4.94,4.94){-179.79}{-0.21}\closepath}}
\rput(45,0.25){}
\rput(40.11,7.09){}
\rput(40.1,0.51){}
\rput(40,13.25){$e_b$}
\psline(37.5,4.83)(23.75,11.05)
\rput{0}(40,11.05){\pscustom[]{\psellipticarc(0,0)(5,5){-0.57}{180.57}\closepath}}
\rput(40.05,2.38){$\sigma$}
\rput{0}(40.05,4.72){\pscustom[]{\psellipticarc(0,0)(4.94,4.94){-179.67}{-0.33}\closepath}}
\psline(23.75,4.8)(37.5,11)
\rput(59.73,7.12){}
\rput(59.72,0.53){}
\psline(62.5,4.8)(62.5,11.05)
\rput(59.68,2.4){$\tau$}
\rput{0}(59.68,4.74){\pscustom[]{\psellipticarc(0,0)(4.95,4.95){-179.73}{-0.27}\closepath}}
\psline(42.5,4.8)(57,11)
\psline(56.88,4.8)(42.5,11.05)
\psline[linewidth=0.25,fillcolor=blue,fillstyle=solid](56.88,10.77)(56.88,13.16)
\psline[linewidth=0.25,fillcolor=blue,fillstyle=solid](55,13.12)(58.75,13.12)
\psline[linewidth=0.25,fillcolor=blue,fillstyle=solid](55.94,14.11)(57.81,14.11)
\psline[linewidth=0.25,fillcolor=blue,fillstyle=solid](56.64,15.09)(57.11,15.09)
\psline[linewidth=0.25,fillcolor=blue,fillstyle=solid](55.47,13.61)(58.28,13.61)
\psline[linewidth=0.25,fillcolor=blue,fillstyle=solid](56.41,14.6)(57.34,14.6)
\psline[linewidth=0.25,fillcolor=blue,fillstyle=solid](62.5,10.77)(62.5,13.16)
\psline[linewidth=0.25,fillcolor=blue,fillstyle=solid](60.62,13.12)(64.38,13.12)
\psline[linewidth=0.25,fillcolor=blue,fillstyle=solid](61.56,14.11)(63.44,14.11)
\psline[linewidth=0.25,fillcolor=blue,fillstyle=solid](62.27,15.09)(62.73,15.09)
\psline[linewidth=0.25,fillcolor=blue,fillstyle=solid](61.1,13.61)(63.9,13.61)
\psline[linewidth=0.25,fillcolor=blue,fillstyle=solid](62.03,14.6)(62.96,14.6)
\rput(68,9){,}
\end{pspicture}

\]
\[
\ifx\JPicScale\undefined\def\JPicScale{1}\fi
\psset{unit=\JPicScale mm}
\psset{linewidth=0.3,dotsep=1,hatchwidth=0.3,hatchsep=1.5,shadowsize=1,dimen=middle}
\psset{dotsize=0.7 2.5,dotscale=1 1,fillcolor=black}
\psset{arrowsize=1 2,arrowlength=1,arrowinset=0.25,tbarsize=0.7 5,bracketlength=0.15,rbracketlength=0.15}
\begin{pspicture}(0,0)(66.88,16.25)
\rput(-5.62,8.12){$P'(a,b)=\sum_c P'(a,b,c)=$}
\rput(21.11,7.35){}
\rput(21.1,0.76){}
\rput(21,13.5){$e_a$}
\psline(17.5,5)(17.5,11.25)
\rput{90}(20.63,11.2){\pscustom[]{\psellipticarc(0,0)(5.05,5){-90}{90}\closepath}}
\rput(21.05,2.63){$\rho$}
\rput{0}(20.75,4.64){\pscustom[]{\psellipticarc(0,0)(4.95,4.64){176.54}{363.46}\closepath}}
\rput(45,0.45){}
\rput(40.11,7.29){}
\rput(40,13.45){$e_b$}
\psline(30.62,5)(23.75,11.25)
\rput{90}(40,11.2){\pscustom[]{\psellipticarc(0,0)(5.05,5){-90}{90}\closepath}}
\psline(23.75,5)(37.5,11.2)
\rput(59.73,7.32){}
\rput(59.72,0.73){}
\psline(62.5,5)(62.5,11.25)
\rput(59.68,2.61){$\tau$}
\rput{0}(59.68,4.64){\pscustom[]{\psellipticarc(0,0)(4.95,4.64){176.54}{363.46}\closepath}}
\psline(50,5)(56.88,11.25)
\psline(56.88,5)(42.5,11.25)
\psline[linewidth=0.25,fillcolor=blue,fillstyle=solid](44.38,5.06)(44.38,7.45)
\psline[linewidth=0.25,fillcolor=blue,fillstyle=solid](42.5,7.42)(46.25,7.42)
\psline[linewidth=0.25,fillcolor=blue,fillstyle=solid](43.44,8.4)(45.31,8.4)
\psline[linewidth=0.25,fillcolor=blue,fillstyle=solid](44.14,9.38)(44.61,9.38)
\psline[linewidth=0.25,fillcolor=blue,fillstyle=solid](42.97,7.9)(45.78,7.9)
\psline[linewidth=0.25,fillcolor=blue,fillstyle=solid](43.91,8.89)(44.84,8.89)
\rput(38.85,0.71){}
\rput(33.12,2.5){$\sigma$}
\rput{0}(33.25,4.64){\pscustom[]{\psellipticarc(0,0)(4.95,4.64){176.54}{363.46}\closepath}}
\rput(53.22,0.71){}
\rput(47.5,2.5){$\sigma$}
\rput{0}(47.63,4.64){\pscustom[]{\psellipticarc(0,0)(4.95,4.64){176.54}{363.46}\closepath}}
\psline[linewidth=0.25,fillcolor=blue,fillstyle=solid](36.25,5.06)(36.25,7.45)
\psline[linewidth=0.25,fillcolor=blue,fillstyle=solid](34.38,7.42)(38.12,7.42)
\psline[linewidth=0.25,fillcolor=blue,fillstyle=solid](35.31,8.4)(37.19,8.4)
\psline[linewidth=0.25,fillcolor=blue,fillstyle=solid](36.02,9.38)(36.48,9.38)
\psline[linewidth=0.25,fillcolor=blue,fillstyle=solid](34.85,7.9)(37.65,7.9)
\psline[linewidth=0.25,fillcolor=blue,fillstyle=solid](35.78,8.89)(36.71,8.89)
\psline[linewidth=0.25,fillcolor=blue,fillstyle=solid](56.88,11.3)(56.88,13.7)
\psline[linewidth=0.25,fillcolor=blue,fillstyle=solid](55,13.66)(58.75,13.66)
\psline[linewidth=0.25,fillcolor=blue,fillstyle=solid](55.94,14.65)(57.81,14.65)
\psline[linewidth=0.25,fillcolor=blue,fillstyle=solid](56.64,15.62)(57.11,15.62)
\psline[linewidth=0.25,fillcolor=blue,fillstyle=solid](55.47,14.15)(58.28,14.15)
\psline[linewidth=0.25,fillcolor=blue,fillstyle=solid](56.41,15.14)(57.34,15.14)
\psline[linewidth=0.25,fillcolor=blue,fillstyle=solid](62.5,11.39)(62.5,13.79)
\psline[linewidth=0.25,fillcolor=blue,fillstyle=solid](60.62,13.75)(64.38,13.75)
\psline[linewidth=0.25,fillcolor=blue,fillstyle=solid](61.56,14.73)(63.44,14.73)
\psline[linewidth=0.25,fillcolor=blue,fillstyle=solid](62.27,15.71)(62.73,15.71)
\psline[linewidth=0.25,fillcolor=blue,fillstyle=solid](61.1,14.23)(63.9,14.23)
\psline[linewidth=0.25,fillcolor=blue,fillstyle=solid](62.03,15.23)(62.96,15.23)
\rput(66.88,8.12){.}
\end{pspicture}

\]
\[
\ifx\JPicScale\undefined\def\JPicScale{1}\fi
\psset{unit=\JPicScale mm}
\psset{linewidth=0.3,dotsep=1,hatchwidth=0.3,hatchsep=1.5,shadowsize=1,dimen=middle}
\psset{dotsize=0.7 2.5,dotscale=1 1,fillcolor=black}
\psset{arrowsize=1 2,arrowlength=1,arrowinset=0.25,tbarsize=0.7 5,bracketlength=0.15,rbracketlength=0.15}
\begin{pspicture}(0,0)(45,10.23)
\rput(30,5){$=1$}
\rput(21.1,0.56){}
\rput(45,0.25){}
\rput(40.1,0.51){}
\rput(19.1,7.66){}
\rput(19.09,1.07){}
\rput(19.05,2.94){$\sigma$}
\rput{0}(19.05,5.28){\pscustom[]{\psellipticarc(0,0)(4.95,4.95){-179.73}{-0.27}\closepath}}
\psline[linewidth=0.25,fillcolor=blue,fillstyle=solid](15.62,5.06)(15.62,7.45)
\psline[linewidth=0.25,fillcolor=blue,fillstyle=solid](13.75,7.41)(17.5,7.41)
\psline[linewidth=0.25,fillcolor=blue,fillstyle=solid](14.69,8.4)(16.56,8.4)
\psline[linewidth=0.25,fillcolor=blue,fillstyle=solid](15.39,9.38)(15.86,9.38)
\psline[linewidth=0.25,fillcolor=blue,fillstyle=solid](14.22,7.9)(17.03,7.9)
\psline[linewidth=0.25,fillcolor=blue,fillstyle=solid](15.16,8.89)(16.09,8.89)
\psline[linewidth=0.25,fillcolor=blue,fillstyle=solid](21.88,5.06)(21.88,7.45)
\psline[linewidth=0.25,fillcolor=blue,fillstyle=solid](20,7.41)(23.75,7.41)
\psline[linewidth=0.25,fillcolor=blue,fillstyle=solid](20.94,8.4)(22.81,8.4)
\psline[linewidth=0.25,fillcolor=blue,fillstyle=solid](21.64,9.38)(22.11,9.38)
\psline[linewidth=0.25,fillcolor=blue,fillstyle=solid](20.47,7.9)(23.28,7.9)
\psline[linewidth=0.25,fillcolor=blue,fillstyle=solid](21.41,8.89)(22.34,8.89)
\end{pspicture}

\]
giving \cref{tri2}. \Cref{tri3} follows similarly.
\end{proof}
\begin{cor}\label{gentobias}
\Cref{tobiasineq} holds whenever $P \in \Gs$ for the GDAG in \cref{triangledag}.
\end{cor}
\begin{proof}
For any probability distribution $\cmutual{A}{C}{B} \geq 0$ and $\cshan{B}{AC} \geq 0$ and so
\begin{equation}
  \mutual{A}{B} + \mutual{B}{C} = \shan{B} + \mutual{A}{C} - \cmutual{A}{C}{B} - \cshan{B}{AC} \leq \shan{B} + \mutual{A}{C}.
\end{equation}
Applying \cref{trithm} we obtain a $P'$ with $\mutual{A}{C} = 0$ so that
\begin{equation}
  \mutual{A}{B} + \mutual{B}{C} \leq \shan{B}.
\end{equation}
But this inequality only involves $P'(a,b)$ and $P'(b,c)$, which equal $P(a,b)$ and $P(b,c)$ respectively, and so this inequality holds for $P$ as well.
\end{proof}

\subsection{The instrumental GDAG}\label{instrsec}
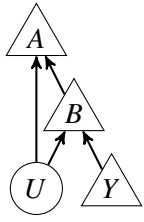
\begin{figure}
\begin{center}
\begin{tikzpicture}
  \node[q](A) at (0,0) {$U$};
  \node[c](B) at (1,0) {$Y$};
  \node[c](C) at (0.5,1) {$B$}
  edge[e] (A)
  edge[e] (B);
  \node[c](D) at (0,2) {$A$}
  edge[e] (A)
  edge[e] (C);
\end{tikzpicture}
\end{center}
\caption{The relevant GDAG for ``instrumental inequalities''.}
\label{instrdag}
\end{figure}

The fact that $\Cs \subsetneq \Is$ for the GDAG in \cref{instrdag} has already been noted in the causality literature \cite{instrument}. The original interest in this DAG arose in the study of cases of imperfect compliance in a controlled trial. 
For example $Y$ might be a randomly assigned treatment, $B$ the treatment the patient actually follows, and $A$ recovery. There could be factors $U$ that influence both the chance of recovery under each treatment, and also the chance of compliance with a particular treatment. This model does not imply any conditional independences on $\{A, B, Y\}$, but in \cite{instrument}, it is shown that it can still be tested because for any $P \in \Cs$,
\begin{equation}
  \max_b \sum_a \max_y P(a,b|y) \leq 1 \label{instrineq}.
\end{equation}

This is known as the \emph{instrumental inequality}. Here we strengthen this result to
\begin{thm}
\Cref{instrineq} holds for any $P \in \Gs$.
\end{thm}
\begin{proof}
Since $P \in \Gs$, there is a bipartite preparation-test at $U$ and choices of observation-test at $A$ and $B$:
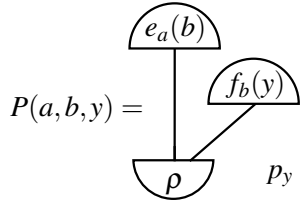
\[
\ifx\JPicScale\undefined\def\JPicScale{1}\fi
\psset{unit=\JPicScale mm}
\psset{linewidth=0.3,dotsep=1,hatchwidth=0.3,hatchsep=1.5,shadowsize=1,dimen=middle}
\psset{dotsize=0.7 2.5,dotscale=1 1,fillcolor=black}
\psset{arrowsize=1 2,arrowlength=1,arrowinset=0.25,tbarsize=0.7 5,bracketlength=0.15,rbracketlength=0.15}
\begin{pspicture}(0,0)(40,29)
\rput(40,8){}
\rput(5,15){$P(a,b,y)=$}
\rput(17,13.82){}
\psline(17.99,10)(18,7.98)
\rput(17.99,5){$\mathcal{\rho}$}
\rput(16.99,6){}
\rput{0}(17.81,8.31){\pscustom[]{\psellipticarc(0,0)(5.32,5.32){-177.98}{-2.02}\closepath}}
\rput(28.75,18.12){$f_b(y)$}
\psline(20,8)(28.75,15.62)
\rput(32,6){$p_y$}
\rput(18.12,25.62){$e_a(b)$}
\psline(17.99,8)(18,23)
\rput{0}(18,23){\pscustom[]{\psellipticarc(0,0)(6,6){0}{180}\closepath}}
\rput{0}(28.5,15.62){\pscustom[]{\psellipticarc(0,0)(6,6){0}{180}\closepath}}
\end{pspicture}
 
\]
These can be used to define a no-signalling distribution $P'(a, b|x,y)$ such that $P(a,b|y) = P'(a,b|x=b,y)$. Using no-signalling from $y$ to $a$ we can write $P'(a,b|x,y) = P'(a|x)P'(b|a,x,y)$. We can now adapt the proof in \cite{instrument} as follows.

For each $(a,b)$, define $y(a,b)$ as the choice of $y$ the maximizes $P(a,b|y)$. Then
\begin{equation}
  \sum_a P(a,b|y(a,b)) = \sum_a P'(a,b|x=b,y(a,b)) = \sum_a P'(a|x=b) P'(b|a,x=b,y(a,b)).
\end{equation}
Certainly $P'(b|a,x,y) \leq 1$, and the final term above is a convex combination of such, and so
\begin{equation}
\sum_a P(a,b|y(a,b)) \leq 1.
\end{equation}
Recalling the definition of $y(a,b)$ this is exactly
\begin{equation}
\sum_a \max_y P(a,b|y) \leq 1.
\end{equation}
Since this holds for all $b$ we have \cref{instrineq}.
\end{proof}
Since there are no observable independences for this GDAG, $\Is$ is just the set of all probability distributions. 
Hence this result establishes that $\Gs \subsetneq \Is$.

\section{Towards a classification of ``interesting'' GDAGs}\label{sec:towards}

\begin{figure}
\begin{center}
\begin{tikzpicture}
  \node[q](A) at (1,0){};
  \node[c](B) at (0,0){$X$};
  \node[c](D) at (1.5,1){$B$}
  edge[e] (A);
  \node[c](E) at (0.5,1){$A$}
  edge[e] (A)
  edge[e] (B);
\end{tikzpicture}
\end{center}
\caption{A bipartite Bell scenario where only Alice has a choice of measurement.}
\label{boringbell}
\end{figure}
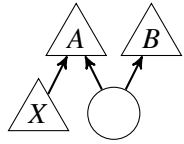

It is known that a Bell scenario where only one party has a choice of measurement (\cref{boringbell}) is not ``interesting''. What exactly does this mean? Certainly it doesn't mean that there are no restrictions on the probability distributions: there is still no-signalling from Alice to Bob. However, this is a conditional independence $X \indep B$ which follows from $d$-separation. Hence, by definition, it is satisfied by all distributions in $\Is$. The reason this scenario is not interesting is that even for classical distributions there are no further restrictions, i.e. $\Cs = \Is$.

Since we have seen that for any GDAG $\Cs \subseteq \Qs \subseteq \Gs \subseteq \Is$, GDAGs in which $\Cs = \Is$ must have $\Cs = \Qs = \Gs = \Is$. Hence there is very little to say about such GDAGs except for listing the observable conditional independences. It is therefore of interest to classify which GDAGs have $\Cs = \Is$ and which do not. The GDAGs that do not are then candidates for quantum advantages in (``black-box'') information processing, settings to compare quantum theory to more general theories, and so on.

Here we make significant progress towards such a classification by providing a sufficient condition for $\Cs = \Is$ and providing strong evidence that our condition is also necessary, at least for small GDAGs. This classification problem may be of interest even for purely classical causal inference, since if one has a candidate causal structure for which $\Cs \subsetneq \Is$ then it can be ruled out by checks that go beyond observable conditional independences (like Bell inequalities). On the other hand, if a candidate causal structure has $\Cs = \Is$ then checking the observable conditional independences implied by $d$-separation suffices for the existence of a (classical) model.

\subsection{A sufficient condition for $\Cs = \Is$}\label{sufficientboring}
We begin by observing that certain changes to a GDAG can only make $\Cs$ smaller. We will use the notation $X \qpath Y$ to denote the existence of a directed path from a node $X$ to node $Y$, where any intermediate nodes are unobserved.

\begin{thm}\label{trsfthm}
  Consider the set of of classical corelations $\Cs_G$ for a GDAG $G$.
  Suppose that one of the following transformations is performed on $G$, producing a GDAG $H$:
  \begin{enumerate}
    \item\label[trsf]{trsf1} Removal of an edge.
    \item\label[trsf]{trsf2} Removal of an isolated unobserved node.
    \item\label[trsf]{trsf3} Addition of an edge $X \to Y$ where previously $X \qpath Y$.
    \item\label[trsf]{trsf4} Addition of an edge $X \to Y$ where previously $\PA{X} \subseteq \PA{Y}$ and $\PA{X}$ contained at least one unobserved node.
  \end{enumerate}
Then $\Cs_{H}\subseteq \Cs_{G}$.
\end{thm}
These transformations are illustrated in \cref{trsfig}.
\begin{figure}
\[
\begin{matrix}
1:&
\begin{matrix}
\begin{tikzpicture}
  \node[c](X) at (0,0) {$X$};
  \node[c](Y) at (1,0) {$Y$}
  edge[e] (X);
\end{tikzpicture}
\end{matrix}
&\to&
\begin{matrix}
\begin{tikzpicture}
  \node[c](X) at (0,0) {$X$};
  \node[c](Y) at (1,0) {$Y$};
\end{tikzpicture}
\end{matrix}&\qquad&
2:&
\begin{matrix}
\begin{tikzpicture}
  \node[q](X) at (0,0) {$X$};
\end{tikzpicture}
\end{matrix}
&\to&
\\[1em]
3:&
\begin{matrix}
\begin{tikzpicture}
  \node[c](X) at (0,0) {$X$};
  \node[q](Z) at (1,0) {$Z$}
  edge[e] (X);
  \node[c](Y) at (1,1) {$Y$}
  edge[e] (Z);
\end{tikzpicture}
\end{matrix}
&\to&
\begin{matrix}
\begin{tikzpicture}
  \node[c](X) at (0,0) {$X$};
  \node[q](Z) at (1,0) {$Z$}
  edge[e] (X);
  \node[c](Y) at (1,1) {$Y$}
  edge[e] (Z)
  edge[e] (X);
\end{tikzpicture}
\end{matrix}&&
4:&
\begin{matrix}
\begin{tikzpicture}
  \node[q](Z) at (0,0) {$Z$};
  \node[c](X) at (-0.5,1) {$X$}
  edge[e] (Z);
  \node[c](Y) at (0.5,1) {$Y$}
  edge[e] (Z);
\end{tikzpicture}
\end{matrix}
&\to&
\begin{matrix}
\begin{tikzpicture}
  \node[q](Z) at (0,0) {$Z$};
  \node[c](X) at (-0.5,1) {$X$}
  edge[e] (Z);
  \node[c](Y) at (0.5,1) {$Y$}
  edge[e] (X)
  edge[e] (Z);
\end{tikzpicture}
\end{matrix}\\
\end{matrix}
\]
\caption{Illustrations of the allowed transformations in \cref{trsfthm}.}
\label{trsfig}
\end{figure}
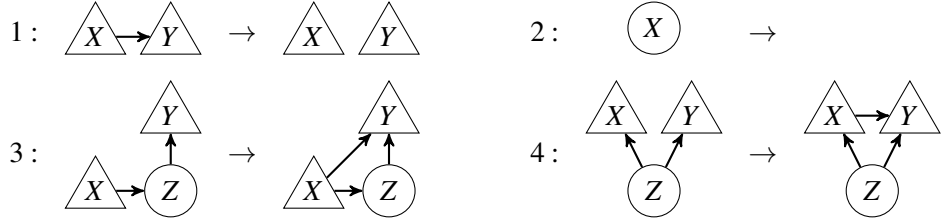

\begin{proof}
  We need to prove that if $P\in \Cs_H$ (i.e.~$P$ is classical for the new GDAG $H$) then $P\in\Cs_G$ (i.e.~$P$ is classical for the old GDAG $G$).
  We shall use the fact that $P$ is classical for a GDAG if and only if there exists a functional causal model for $P$ using the underlying DAG \cite{pearl}.
  In a functional causal model, the value of each node $Z$ is given by a function $z=f(\pa{z}, n_z)$ of its parents and a noise variable, and the noise variables are independently distributed.
  For each transformation, we shall show that if a functional causal model exists for $P$ defined on $H$, then a functional causal model exists for $P$ defined on $G$:
\begin{enumerate}
	\item  
	In $H$ the argument to a function has been removed, e.g. if a node $Z$ has parents $X$ and $Y$, then $z = f(x,y,n_z)$ becomes $z = f'(x,n_z)$. 
	We can define a functional causal model for $G$ using the one for $H$ by allowing the function to trivially depend on its new argument, e.g. $f(x,y,n_z)=f'(x,n_z)$. 
	By definition, this gives the same probabilities for all nodes.
	\item 
	We can define a model for $G$ by giving the isolated node $Z$ an arbitrary error variable $N_Z$ and making $Z$ an arbitrary function of it. 
	This has no effect on the probabilities for any other variable, which includes all the observable variables.
	\item 
	In both $G$ and $H$ we have $X \qpath Y$, but in $H$ we also have $X\to Y$.
	To define a model for $G$ we must absorb the dependence of $Y$ on $X$ that exists for $H$. 
	We can do so by using the unobserved nodes $Z\s{i}$ in the path $X\qpath Y$. 
	Specifically,  
	for each of the random variables $Z\s{i}$ defined for $H$, we define an `enlarged' variable $W\s{i}$ that includes a copy of $X$, when defining a model for $G$.
	That is, $z = f(\pa{z},n_z)$ becomes $w:=(z, x) = (f(\pa{z},n_z), x)$. 
	We then replace the dependence of the function at $Y$ on $X$ by its copy in $W$, i.e. $y = f(z,x,n_y)$ becomes $y = f(z', n_y)$. 
	This procedure does not affect any of the observable probabilities.
	\item 
	In $H$, the variable $Y$ is now a function of $X$. 
	In turn, $X$ is a function of its parents and an error variable $N_X$. 
	But since 
	$\PA{X}\subseteq\PA{Y}$,
	to define a model for $G$ we need only 
	ensure the dependence of $Y$ on $N_X$.
	Since $N_X$ is independently distributed, we can move this into an unobserved parent of $X$, say $Z$, which exists by assumption. 
	Specifically, we define $z':=(z,n_x)$, and then
	$x = f(z,n_x)$ for $H$ becomes $x = f'(z'):=f(z,n_x)$ for $G$. 
	We let $Y$ be calculated as before, but in place of the direct dependence on $X$, we use the same function used to calculate $x$ at $X$, e.g. $y = g(x,z,n_y)$ becomes 
	$y = g'(z', n_y):=g(f(z,n_x),z,n_y)$.
	The only variable whose probabilities have been changed is $Z$, which is not observable.\qedhere
\end{enumerate}
\end{proof}

The sufficient condition for $\Is = \Cs$ is as follows. If starting with a given GDAG one can apply a sequence of the above transformations and produce a GDAG with:
\begin{enumerate}
  \item no unobserved nodes, and
  \item requiring no more conditional independences on the observed nodes than the original GDAG did,
\end{enumerate}
then $\Is = \Cs$ for the original GDAG. To see this, start with some probability distribution in $\Is$. Recalling that the conditional independences are the only restrictions on (G)DAGs with no latent variables, the above two properties ensure that the distribution is classical for the new GDAG. But then by repeated applications of \cref{trsfthm} there is a classical model for the original GDAG with the same probabilities for the observed nodes, and we are done.

For example, this condition establishes that the Bell scenario with only one setting, \cref{boringbell}, indeed has $\Is = \Cs$, as shown in \cref{boringbellproof}.
\begin{figure}
\[
\begin{matrix}
\begin{matrix}
\begin{tikzpicture}
  \node[q](A) at (1,0){};
  \node[c](B) at (0,0){$X$};
  \node[c](D) at (1.5,1){$B$}
  edge[e] (A);
  \node[c](E) at (0.5,1){$A$}
  edge[e] (A)
  edge[e] (B);
\end{tikzpicture}
\end{matrix}
&\underset{\ref{trsf4}}{\to}&
\begin{matrix}
\begin{tikzpicture}
  \node[q](A) at (1,0){};
  \node[c](B) at (0,0){$X$};
  \node[c](D) at (1.5,1){$B$}
  edge[e] (A);
  \node[c](E) at (0.5,1){$A$}
  edge[e] (A)
  edge[e] (B)
  edge[e] (D);
\end{tikzpicture}
\end{matrix}
&\underset{\ref{trsf1}}{\to}&
\begin{matrix}
\begin{tikzpicture}
  \node[q](A) at (1,0){};
  \node[c](B) at (0,0){$X$};
  \node[c](D) at (1.5,1){$B$};
  \node[c](E) at (0.5,1){$A$}
  edge[e] (B)
  edge[e] (D);
\end{tikzpicture}
\end{matrix}
&\underset{\ref{trsf2}}{\to}&
\begin{matrix}
\begin{tikzpicture}
  \node[c](B) at (0,0){$X$};
  \node[c](D) at (1.5,1){$B$};
  \node[c](E) at (0.5,1){$A$}
  edge[e] (B)
  edge[e] (D);
\end{tikzpicture}
\end{matrix}
\end{matrix}
\]
\caption{Repeated applications of \cref{trsfthm} transform the GDAG of \cref{boringbell} into a new GDAG without enlarging $\Cs$. (The numbers under the arrows indicate the relevant transformation from \cref{trsfthm}.) Allowed distributions on the final GDAG are constrained only by $X \indep B$, which held for $\Is$ in the initial GDAG, and so $\Cs = \Is$ in the initial GDAG. }
\label{boringbellproof}
\end{figure}
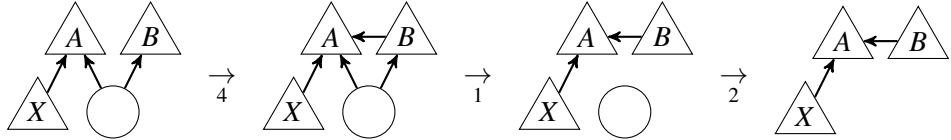

\subsection{Results for small GDAGs}
\begin{table}
\begin{center}
\begin{tabular}{cccc}
  Nodes & Number of GDAGs & Number for which our condition holds & Percent \\
  1 & 2        & 2        & 100\%  \\
  2 & 7        & 7        & 100\%  \\
  3 & 40       & 40       & 100\%  \\
  4 & 420      & 419      & 99.8\% \\
  5 & 8628     & 8532     & 98.9\% \\
  6 & 357468   & 347287   & 97.2\% \\
  7 & 29989052 & 28370373 & 94.6\%
\end{tabular}
\end{center}
\caption{The results of our condition for GDAGs of size 1 to 7. It is plausible that the fraction of GDAGs for which $\Cs = \Is$ tends to zero as the number of nodes tends to infinity, because larger and larger GDAGs should be more and more likely to contain, for example, a Bell scenario.}
\label{smallresults}
\end{table}

Using a strategy described in \cref{searchstrat}, and algorithms from \cite{Studeny1998} to keep track of conditional independeces, we have searched for applications of the above condition to all GDAGs with up to seven nodes. The results are shown in \cref{smallresults}. Our condition is powerful enough to show that the overwhelming majority of small GDAGs have $\Cs = \Is$. Indeed this is the case for all GDAGs of size three or smaller, and the only GDAG of size four is that of \cref{instrsec} for which it was already known that $\Cs \subsetneq \Is$. The 96 GDAGs of size 5 to which our condition does not apply are mostly trivial modifications of that of \cref{instrsec}, for which the proof that $\Cs \subsetneq \Is$ will easily carry over. To eliminate such GDAGs from consideration we developed a number of reduction criteria. For completeness these are described in \cref{reduction}.

Once these reduction criteria have been applied, there remain 2 GDAGs of size five and 18 of size six. If we can show that these 20 GDAGs have $\Cs \subsetneq \Is$ then we will have shown that our necessary condition is also sufficient, at least for GDAGs of size six or less. A full characterisation of $\Cs$ in a general scenario is not known, however necessary conditions for membership of $\Cs$ can be derived by searching for ``Shannon-type entropic inequalities''. These are linear inequalities expressed purely in terms of the Shannon entropy $\shan{X}$ of subsets of variables. See \cite{chaves} and references therein for the details of this approach.

For each GDAG we construct the Shannon cone (defined by the positivity of all conditional mutual informations) for all variables, observable and latent. For each node $X$ we add
\[ \cmutual{X}{\text{non-descendants of $X$}}{\PA{X}} \leq 0 \]
to enforce the Markov condition. Finally we use Fouirer-Motzkin elimination to project out entropies involving the latent variables. This gives a set of entropic inequalities $E_\Cs$.

The resulting inequalities are necessary conditions for membership of $\Cs$. However, we are interested in comparing $\Cs$ with $\Is$. Hence we repeat the process for $\Is$. We start with the Shannon cone on the observable variables, and add $\cmutual{X} {Y}{Z} \leq 0$ whenever $X$ and $Y$ are $d$-separated by $Z$. This gives a second set of entropic inequalities $E_\Is$.

For 19 of the 20 GDAGs we find inequalities in $E_\Cs$ that do not follow from those of $E_\Is$. Unless the inequality is a non-Shannon-type inequality for $\Is$, this establishes that $\Cs \subsetneq \Is$. Since non-Shannon-type inequalities rarely play a role, this is rather good evidence. For most of the GDAGs it is straightforward to find explicit $P \in \Is$ that violate one of $E_\Cs$, thus definitively establishing $\Cs \subsetneq \Is$. Curiously, the one GDAG for which $E_\Is = E_\Cs$ is the bipartite Bell scenario. Fortunately, we already know that $\Cs \subsetneq \Is$ for that case! The GDAGs and corresponding inequalities are listed in \cref{dagtable}.

These results provide excellent evidence that our sufficient condition for $\Cs = \Is$ is also necessary for all GDAGs with six or fewer nodes. Perhaps it is in fact necessary for an arbitrary GDAGs.

\section{Conclusions}
Here we have proposed a way to combine the frameworks of generalised probabilistic theories and causal Bayesian networks. We believe that the results we have obtained suggest that this proposal is worth exploring further, although the two fundamentally distinct types of node mean it is unlikely to be the final word on non-classical causation.

Our first main result was that the graphical $d$-separation criteria for conditional independence remains sound for generalised networks. This should be useful, since the classical soundness result is very fundamental to the classical theory. For example, the main algorithm for causal inference in the presence of latent variables, IC*, uses only observable conditional independences. Hence it will still operate correctly in our generalisation. It would be worth exploring similar generalisations of other fundamental parts of the classical theory, for example the criteria for two causal structures to have the same observable consequences.

We then found that some other constraints on observed probabilities also generalise to this setting.
This shows that even in its weakest interpretation, causal structure has more interesting consequences than ``no signalling'' in the Bell scenario, even extended to include all observable conditional independences.  This has interesting foundational consequences.  If the violation of Bell inequalities is to be explained by accepting altered causal structure, one must give up hope of an explanation of observed conditional independences such as no-signalling based on causal structure \cite{wood}. We now see that there are other, more intricate, limitations which would also be left unexplained by an altered causal structure. Since our techniques for finding such limits were rather ad hoc, the main open problem here is to obtain a more systematic understanding of these constraints. The entropic inequalities look like a promising place to start: indeed we do not know of any example of such an inequality being violated by any generalised probabilistic theory.

Finally, we have considered the problem of identifying whether or not the only consequences of a GDAG are conditional independences, i.e. $\Cs = \Is$. We have presented a sufficient condition. Proving the necessity of this (or any other) condition would shed light on the conceivable forms of ``device-independent'' non-classicality. If a GDAG has $\Cs \subsetneq \Is$, then one could also ask more fine-grained questions: is $\Cs \subsetneq \Qs$ (quantum non-classicality), $\Qs \subsetneq \Gs$ (post-quantum correlations), $\Gs \subsetneq \Is$ (theory-independent limits on correlation)? Other interesting classification problems include understanding when the distributions on some nodes, conditioned on some others, form a convex set.

Thinking more widely, it is of great interest to ask about the extent to which classical causal principles, such as Reichenbach's principle, can be extended to the unobserved nodes \cite{Cavalcanti:2013,neutral,Henson:2005wb}, and also whether quantum mechanics supports a stronger analogy to such classical principles than other GPTs (see \textit{e.g.} \cite{Spekkens:2012,Dowker:2013}). But before these deep issues are tackled, it is important to understand what causal features of classical theories carry over to the most general cases.  This work addresses the latter issue, and it is our hope that these other issues of causality in quantum mechanics, and beyond, can also be fruitfully explored using our framework.

\paragraph{Acknowledgements.}
We are grateful for useful discussions with Jonathan Barrett, Giulio Chiribella, Tobias Fritz, Anirudh Krishna and Rob Spekkens. Research at Perimeter Institute is supported in part by the Government of Canada through NSERC and by the Province of Ontario through MRI.  Work by JH and RL is supported by grants from the John Templeton Foundation.  JH also receives support form EPSRC grant \textit{DIQIP} and ERC grant \textit{NLST}.

\printbibliography
\appendix
\section{Comparison with other approaches}\label{sec:comparison}

Recent work by other authors has also considered correlations on general causal structures. We shall restrict our focus to those approaches which have been specifically used to study classical correlations resulting from quantum processes on general causal structures. 
Hence we omit works that give a quantum version of Bayesian networks by replacing probabilities with amplitudes (e.g. \cite{tucci}), or that only apply to states at a single time-step (e.g. \cite{poulin}), since neither appears to support the causal interpretation which we are interested in. More relevant are the ``Quantum Causal Networks'' of \cite{laskey}, but these are difficult to compare to our approach since they treat entanglement as a new type of causal relation indicated by an undirected edge, whereas in our approach entanglement requires an analog of a ``common cause,'' that is, mutual ancestors. Most closely related are two lines of work, based on source-measurement hypergraphs and circuit DAGs respectively. The idea of having two different types of node, and specifically the choice of triangles and circles, comes from a more general project to recast quantum theory as a theory of inference. To aid the reader who has encountered any of these three approaches, here we compare their definitions with ours.

\subsection{Hypergraphs}
Building on the idea of ``$N$-locality'' from \cite{bilocal}, in \cite{tobias} a causal structure is represented by a hypergraph, with vertices representing measurements and edges representing sources. This can be translated into our formalism by turning each vertex into an observed node, and each hypergraph edge into an unobserved node with an edge going to every member of the hypergraph edge. What is called a ``correlation'' in \cite{tobias} then agrees with our definition of a member of $\Is$, and the definitions of classical and quantum correlations map directly to our definitions of $\Cs$ and $\Qs$.

This close translation means that some of our results touch directly on the results and open problems in \cite{tobias}. Our triangle result answers the first part of Problem 3.4 in \cite{tobias} in the negative. Our investigation in \cref{sec:towards} seeks to address (a generalisation of) Problem 3.6 in \cite{tobias}. For example, the criteria given in \cref{sufficientboring} enables a graphical proof of the ``if'' part of Theorem 3.8 in \cite{tobias}, see \cref{tobias2us}.

Many GDAGs in our formalism will not correspond to any hypergraph in the formalism of \cite{tobias}. For example, the GDAG in \cref{instrsec} cannot be represented as a hypergraph as there is no way to encode the edge from $B$ to $A$.

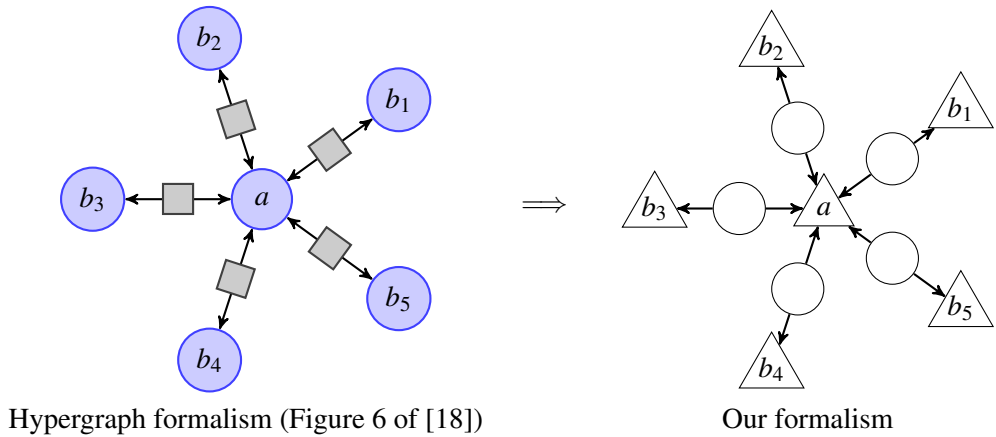
\begin{figure}
\begin{center}
\begin{tabular}{ccc}
\begin{tabular}{c}
\begin{tikzpicture}[node distance=2.1cm,>=stealth',thick,scale=0.75]
\tikzstyle{place}=[circle,thick,draw=blue!75,fill=blue!20,minimum size=8mm]
\tikzstyle{transition}=[rectangle,thick,draw=black!75,fill=black!20,minimum size=4.0mm]
\node[place] (a) at (0,0) {$a$} ;
\node[place] (b1) at (36:3) {$b_1$} ;
\node[place] (b2) at (108:3) {$b_2$} ;
\node[place] (b3) at (180:3) {$b_3$} ;
\node[place] (b4) at (252:3) {$b_4$} ;
\node[place] (b5) at (324:3) {$b_5$} ;
\draw[<->] (a) -- (b1) node [midway,sloped,above=-6pt,transition] {} ;
\draw[<->] (a) -- (b2) node [midway,sloped,above=-6pt,transition] {} ;
\draw[<->] (a) -- (b3) node [midway,sloped,above=-6pt,transition] {} ;
\draw[<->] (a) -- (b4) node [midway,sloped,above=-6pt,transition] {} ;
\draw[<->] (a) -- (b5) node [midway,sloped,above=-6pt,transition] {} ;
\end{tikzpicture}
\end{tabular}&
$\implies$&
\begin{tabular}{c}
\begin{tikzpicture}[scale=0.75]
\node[q] (q1) at (36:1.5) {};
\node[q] (q2) at (108:1.5) {} ;
\node[q] (q3) at (180:1.5) {} ;
\node[q] (q4) at (252:1.5) {} ;
\node[q] (q5) at (324:1.5) {} ;
\node[c] (a) at (0,0) {$a$}
edge[e] (q1)
edge[e] (q2)
edge[e] (q3)
edge[e] (q4)
edge[e] (q5);
\node[c] (b1) at (36:3) {$b_1$}
edge[e] (q1);
\node[c] (b2) at (108:3) {$b_2$}
edge[e] (q2);
\node[c] (b3) at (180:3) {$b_3$}
edge[e] (q3);
\node[c] (b4) at (252:3) {$b_4$}
edge[e] (q4);
\node[c] (b5) at (324:3) {$b_5$}
edge[e] (q5);
\end{tikzpicture}
\end{tabular}
\\
Hypergraph formalism (Figure 6 of \cite{tobias}) & & Our formalism
\end{tabular}
\end{center}
\caption{In the formalism of \cite{tobias} a causal structure is formally represented by a hypergraph, although the edges are suggestively drawn as squares with arrows to members. To convert to our formalism, an edge becomes an actual (unobserved) node with edges to each member. An application of \cref{sufficientboring} immediately shows that in this ``star'' scenario $\Cs = \Is$.}
\label{tobias2us}
\end{figure}

\subsection{Circuit diagrams}
The ubiquitous circuit diagrams used in quantum computing \cite{nc} and discussions of generalised probabilistic theories (e.g. \cite{Chiribella2010}) can be viewed as DAGs, and seem to suggest a causal interpretation (see \cite{Blute2000} and references therein). Recently this idea has been used specifically for the purpose of exploring Bell-like scenarios \cite{tobias2}.

In \cite{tobias2} a causal structure is represented as a DAG. Hence there is only one type of node, which is always associated with a random variable. Any edge can carry ``hidden variables'' in the classical case or quantum systems in the quantum case. Hence to translate to our formalism, first represent every node as an unobserved node. Then add a supplementary observed node for each of those nodes, and an edge from the unobserved to the supplementary observed node, as in \cref{tobiastwo2us}. Again the definitions of correlation, classical correlation and quantum correlation appear to coincide with $\Is$, $\Cs$ and $\Qs$ respectively (except that \cite{tobias2} allows infinite-valued latent variables, which as already noted may or may not result in more classical correlations). $\Qs$ only matches because every quantum instrument can be replaced by a channel\footnote{A channel is a quantum instrument with only one outcome, i.e. a completely positive trace preserving map.} with an additional ``flag'' system in the output which can later be measured to obtain the result. Finally, \cite{tobias2} considers $\mathtt{C}$-correlations for certain categories $\mathtt{C}$. This is closely related to the CDP formalism of operational-probabilistic theories and so ranging over all $\mathtt{C}$ should, under the above translation, agree with our $\Gs$.

Again many GDAGs in our formalism will not correspond to any DAG in \cite{tobias2}. For example, in the formalism of \cite{tobias2} there is no way to enforce that the edge from $B$ to $A$ in the GDAG of \cref{instrsec} does not carry hidden variables or quantum systems, rather than just the value $b$ as in our formalism.

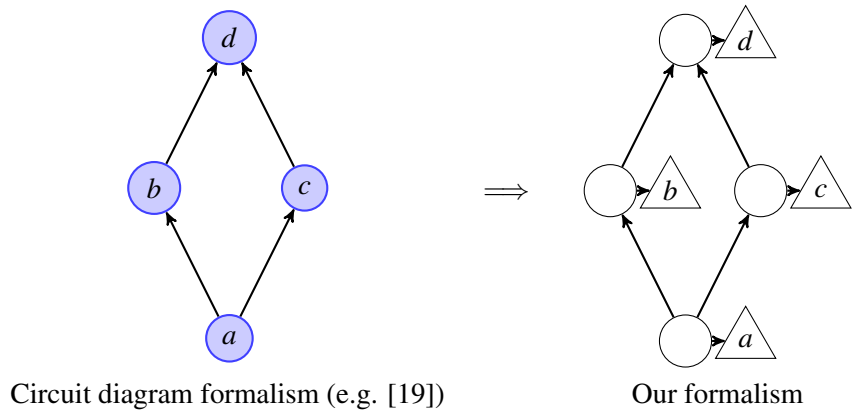
\begin{figure}
\begin{center}
\begin{tabular}{ccc}
\begin{tabular}{c}
\begin{tikzpicture}[node distance=2.3cm,>=stealth',thick]
\tikzstyle{place}=[circle,thick,draw=blue!75,fill=blue!20,minimum size=4mm]
\node[place] (l) at (1,0) {$a$} ;
\node[place] (la) at (0,2) {$b$} ;
\node[place] (lb) at (2,2) {$c$} ;
\node[place] (a2) at (1,4) {$d$} ;
\draw[->] (l) -- (la) ;
\draw[->] (l) -- (lb) ;
\draw[->] (la) -- (a2) ;
\draw[->] (lb) -- (a2) ;
\end{tikzpicture}
\end{tabular}&
$\implies$&
\begin{tabular}{c}
\begin{tikzpicture}
\node[q] (l) at (1,0) {} ;
\node[q] (la) at (0,2) {}
edge[e] (l);
\node[q] (lb) at (2,2) {}
edge[e] (l);
\node[q] (a2) at (1,4) {}
edge[e] (la)
edge[e] (lb);
\node[c] (cl) at (1.8,0) {$a$}
edge[e] (l);
\node[c] (cla) at (0.8,2) {$b$}
edge[e] (la);
\node[c] (clb) at (2.8,2) {$c$}
edge[e] (lb);
\node[c] (ca2) at (1.8,4) {$d$}
edge[e] (a2);
\end{tikzpicture}
\end{tabular}
\\
Circuit diagram formalism (e.g. \cite{tobias2}) & & Our formalism
\end{tabular}
\end{center}
\caption{In the formalism of \cite{tobias2} a causal structure is represented by a DAG. Every edge gets a hidden variable in the classical case and a quantum system in the quantum case, so to represent such a structure in our formalism each node should become two nodes, one of each type, as shown.}
\label{tobiastwo2us}
\end{figure}

\subsection{Quantum theory as a theory of inference}
In \cite{neutral}, Leifer and Spekkens also use GDAGs depicted using circular and triangular nodes. We deliberately use the same notation here, although the approaches are significantly different. 
The aim in \cite{neutral} is to generalise the quantum formalism to the point that one can, for example, talk about the joint quantum state of $A$ and $B$ even if $A$ is the input to a channel and $B$ the output. 
Here we stick to the standard quantum formalism, with tensor products only across space, and limit ourselves to the joint probabilities of the variables on the observed nodes---i.e., the classical variables. 
In \cite{neutral}, the state of a set of triangular nodes is diagonal in a fixed basis and hence encodes a joint probability distribution. 
We use the same notation because we expect the possible sets of joint distributions in \cite{neutral} to match our \Qs.

The main reason that the distributions may not be identical is that when an unobserved node has multiple outgoing edges, we associate a Hilbert space to each edge, giving an explicit tensor product structure. 
In \cite{neutral}, a single Hilbert space is associated with the circular node itself. The meaning of edges is to be in terms of some planned generalisation the classical Markov condition to quantum states. 
Presumably our tensor products will satisfy this condition (see \cref{us2ls} for an example of the likely translation), but there may be quantum states that are ``Leifer-Spekkens Markov'' for a GDAG and yet cannot be expressed using our tensor product form.

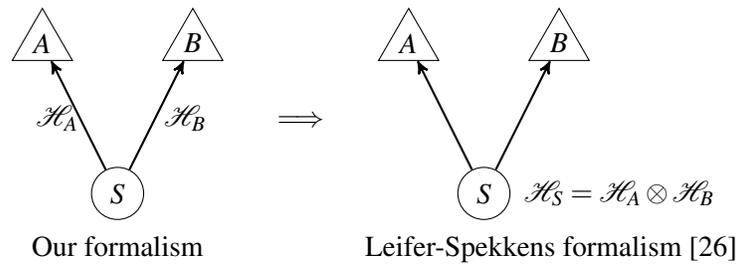
\begin{figure}
\begin{center}
\begin{tabular}{ccc}
\begin{tabular}{c}
\begin{tikzpicture}
  \node[q](A) at (1,0){$S$};
  \node[c](D) at (2,2){$B$}
  edge[e] (A);
  \node[c](E) at (0,2){$A$}
  edge[e] (A);
  \node at (0.2,1) {$\Hb_A$};
  \node at (1.9,1) {$\Hb_B$};
\end{tikzpicture}\end{tabular}&
$\implies$&
\begin{tabular}{c}
\begin{tikzpicture}
  \node[q](A) at (1,0){$S$};
  \node[c](D) at (2,2){$B$}
  edge[e] (A);
  \node[c](E) at (0,2){$A$}
  edge[e] (A);
  \node at (2.8,0) {$\Hb_S = \Hb_A \otimes \Hb_B$};
\end{tikzpicture}
\end{tabular}
\\
Our formalism & & Leifer-Spekkens formalism \cite{neutral}
\end{tabular}
\end{center}
\caption{In our formalism a quantum model for this GDAG consists of two Hilbert spaces, a bipartite quantum state and a POVM on each Hilbert space. In the Leifer-Spekkens formalism there would be a single Hilbert space for $\Hb_S$ with an associated state, and two POVMs on $\Hb_S$ satisfying some Markov condition. Translating from the first to the second just involves letting $\Hb_S$ be the tensor product of the two Hilbert spaces, keeping the state as it is, and tensoring the POVMs with identities so that they act on the whole of $\Hb_S$. Until the Leifer-Spekkens formalism has been fully worked out it is difficult to say whether translation in the opposite direction will always be possible.}
\label{us2ls}
\end{figure}

\section{Proofs of $d$-separation lemmas}\label{proofsec}
\begin{proof}[Proof of \cref{l:dsep}]
(``If.'')
We must show that every pseudo-path from $X$ to $Y$ intersects $Z$.
Assume for contradiction that there exists a pseudo-path from $X$ to $Y$ that does not intersect $Z$.
A pseudo-path cannot intersect $W$ by definition.
Then, by the assumption that $\{U,V,Z,W\}$ is a partition of $G$, a pseudopath from $X$ to $Y$ that does not intersect $Z$ can only contain elements in $U$ or $V$.  
Such a pseudo-path must at some point contain a pair of sequential elements $a \in U$ and $b \in V$.  
But we have also assumed that $m(U) \cap m(V) \subseteq W$, i.e.~the mutual children of $a$ and $b$ are in $W$. 
But this contradicts the definition of a pseudo-path, for which we must have $m(a)\cap m(b)\not\in W$.
Hence no such sequential pair in a pseudopath can exist, and therefore there are no pseudopaths from $X$ to $Y$ that do not intersect $Z$. 

(``Only if.'')
Notice that $W=G \setminus \ian{X \cup Y \cup Z}$ is as in the definition of $d$-separation, and in particular that $W\cap Z=\emptyset$.
We obtain the required partition of $G$ as follows.
Let $U$ be the union of all pseudo-paths that start at any node in $X$ and finish anywhere in $G$ but without intersecting $Z$.  
By the definition of $U$, we have $X\subseteq U$ and $U\cap Z=\emptyset$. 
By the definition of a pseudo-path, $U\cap W=\emptyset$.
Hence $U$, $W$ and $Z$ are disjoint.
Now define $V := G \setminus (U \cup W \cup Z)$. 
This defines a partition $\{U, V, Z, W\}$ of $G$, with $X \subseteq U$.
Now, by assumption all pseudopaths from $X$ to $Y$ intersect $Z$.
Therefore $Y\cap U=\emptyset$, by the definition of $U$.
Since we also have $Y\cap Z=\emptyset$ and $Y\cap W=\emptyset$, and since $\{U, V, Z, W\}$ is a partition of $G$, we therefore have $Y\subseteq V$.
Finally, suppose that there exist $a\in U$ and $b\in V$ such that $m(a) \cap m(b) \not\subseteq W$.
This defines a  pseudo-path from $a$ to $b$ that does not intersect $Z$. 
But then by the definition of $U$, we have $b\in U$
 which contradicts the fact that $b\in V$, since $U\cap V=\emptyset$.
 Hence we have $m(U) \cap m(V) \subseteq W$.
\end{proof}

\begin{proof}[Proof of \cref{l:part1}]
The GMC condition is
\begin{equation}
\label{e:gprobdist}
P(g) = \prod^{m}_{i=1} \cT_{x\s{i}}(\cpa{ x\s{i} })_{\GPA{ X\s{i}}}^{\GCH{ X\s{i}}},
\end{equation}
and we have
\begin{equation}
\label{e:abzprobdist}
P(g') = \sum_w p(g).
\end{equation}

By assumption $W \subseteq G$ contains its own future.  A node that is maximal with respect to $W$ is thus maximal with respect to $G$.  Consider such a maximal node $X\s{j}$, and the following expression:
\begin{equation}
\label{e:onenode}
\sum_{x\s{j}} \cT_{x\s{j}}(\cpa{ x\s{j} })_{\GPA{ X\s{j}}}^{\GCH{ X\s{j}}}.
\end{equation}
A maximal node has no outgoing systems and so $\GCH{ X\s{j}}$ is in this case empty, so \eqref{e:onenode} is an observation test.  Furthermore it is either already deterministic (if $X\s{j}$ is unobserved), or summing over all outcomes $x\s{j}$ makes it deterministic (if $X\s{j}$ is observed).  For both types of node therefore (\ref{e:onenode}) equals the unique deterministic effect on $\GPA{ X\s{j}}$.  Applying lemma \ref{l:compose_effect},
\begin{equation}
\sum_{x\s{j}} \cT_{x\s{j}}(\cpa{ x\s{j} })_{\GPA{ X\s{j}}}^{\GCH{ X\s{j}}} = \top_{\GPA{ X\s{j}}} = \prod_{X\s{i} \rightarrow  X\s{j} \in \GPA{ X\s{i}}} \top_{X\s{i} \rightarrow  X\s{j}}.
\end{equation}

Summing over $x\s{j}$ in (\ref{e:gprobdist}), noting that the maximality of $X\s{j}$ ensures that $x\s{j}$ appears only in the $i=j$ term, and substituting the above expression for that term we have
\begin{equation}
\label{e:summax}
\sum_{x\s{j}} P(g) =
  \prod_{i \in \{1,...,m\} \backslash j} \cT'_{x\s{i}}(\cpa{ x\s{i} })_{\GPA{ X\s{i}}}^{\GCH{ X\s{i}} \backslash X\s{i} \rightarrow X\s{j}},
\end{equation}
where 
\begin{equation}
\cT'_{x\s{i}}(\cpa{ x\s{i} })_{\GPA{ X\s{i}}}^{\GCH{ X\s{i}} \backslash X\s{i} \rightarrow X\s{j}}
= \cT_{x\s{i}}(\cpa{ x\s{i} })_{\GPA{ X\s{i}}}^{\GCH{ X\s{i}}} \top_{X\s{i} \rightarrow X\s{j}}
\end{equation}
where $\top_{X\s{i} \rightarrow X\s{j}}$ is the unique deterministic effect for the system on the edge $X\s{i}$ to $X\s{j}$. 

The upshot is that marginalising over the outcomes for a maximal element $X\s{j}$ produces a probability distribution that fulfils the GMC for the GDAG with that element removed, $G \backslash X\s{j}$.  Because $W \subseteq G$ contains its own future, this process can be repeated for every element in $W$, and so marginalising over every outcome in $W$ results in a distribution satisfying the GMC for $G \backslash W$.
\end{proof}

\begin{proof}[Proof of \cref{l:part2}]
We can write
\begin{equation}
P(x,y,z) =\sum_{u' v'} P(u,v,z)
\end{equation}
where $U'=U \backslash X$, $V'=V \backslash Y$, and the lowercase versions are the associated outcome variables as before.  Define $Z_1=Z \cap m(U) = m(U) \backslash U$ and $Z_2= Z \backslash Z_1$.  Then (using $U Z_1$ as shorthand for $U \cup Z_1$ and so on),
\begin{align}
P(y,z) &=\sum_{u} \prod_{P \in U Z_1} \cT_{p}(\cpa{ p })_{\GPA{ P }}^{\GCH{ P }}  \sum_{v'} \prod_{Q \in V Z_2} \cT_{q}(\cpa{ q })_{\GPA{ Q }}^{\GCH{ Q }} \\
&=\Bigr( \sum_{u'x} \prod_{P \in U Z_1} \cT_{p}(\cpa{ p })_{\GPA{ P }}^{\GCH{ P }}  \Bigl)  \Bigr( \sum_{v'} \prod_{Q \in V Z_2} \cT_{q}(\cpa{ q })_{\GPA{ Q }}^{\GCH{ Q }} \Bigl) .
\label{dsepfactor1}
\end{align}
The factorisation above follows because the nodes $U$, whose outcome variables $u$ are summed over in the first bracket, do not appear in the second bracket.  A node in $U$ is never a parent of a node in $Z_2$, from the definition of $Z_2$ and $Z_1$ above; it is never a parent of a node in $V$ because of condition \ref{e:new2}.  Conversely, a node in $V'$ is never a parent of a node in $Z_1$ or of a node in $U$ for the same reasons.  It follows trivially that a node in $U$ is not the \textit{child} of a node in $V'$ or vice versa.
This establishes the factorisation (and also that the terms in the brackets correspond to closed circuits and are thus probabilities).  For the same reasons we also have
\begin{align}
P(x,y,z) =\Bigr( \sum_{u'} \prod_{P \in U Z_1} \cT_{p}(\cpa{ p })_{\GPA{ P }}^{\GCH{ P }}  \Bigl)  \Bigr( \sum_{v'} \prod_{Q \in V Z_2} \cT_{q}(\cpa{ q })_{\GPA{ Q }}^{\GCH{ Q }} \Bigl) .
\label{dsepfactor2}
\end{align}
Now $P(x|y,z) = P(x,y,z)/P(y,z)$, and the second terms in \eqref{dsepfactor1} and \eqref{dsepfactor2} will cancel. Since $Y \subseteq V$ this means $P(x|y,z)$ is independent of $y$, establishing the conditional independence of $X$ and $Y$ given $Z$.
\end{proof}

\section{A $\Cs = \Is$ search strategy}\label{searchstrat}
It might appear that one has to attempt a potentially unbounded number of transformations to apply the sufficient condition for $\Cs = \Is$ in \cref{sufficientboring}. Fortunately, if any sequence of transformations exists from a GDAG to one satisfying the criteria given there, then one will be found using the following strategy, as we will show below.

Let $T$ (for ``tricky'') be the set of all observed nodes that have unobserved parents. Let $R$ (for ``root'') be the set of all unobserved nodes that have no unobserved parents. Consider every possible ordering of $T$: $T_1, T_2, \cdots, T_n$, with each element $T_i$ associated with every possible $R_i \in R$, with $R_i \qpath T_i$ 
.  For each possibility, apply the transformations as follows:
\begin{enumerate}
  \item Apply \cref{trsf3} to every pair of nodes with $X \qpath Y$.
  \item For $i$ from $1$ to $n$:
    \begin{enumerate}
      \item Applying \cref{trsf1}, remove any edges from $T_j$ (with $j>i$) to $T_i$, and from any unobserved nodes (except $R_i$) to $T_i$.
      \item Use \cref{trsf4} to add edges from $T_i$ to $T_j$ (with $j>i$) where possible.
    \end{enumerate}
  \item Apply \cref{trsf1} to remove any remaining edges incident on unobserved nodes, then use \cref{trsf2} to remove all the unobserved nodes.
\end{enumerate}

It can be seen that \cref{trsf3} can be applied first, as none of the other transformations can increase its applicability. It might as well be applied ``maximally'' as any unhelpful edges can always be removed later.

It can also be seen that \cref{trsf2} must be applied to all the unobserved nodes at some point, to ensure there are none in the final GDAG, and it can always be applied last, as it cannot increase the applicability of any of the other transformations.

All that remains is to show that the second step makes the best use of \cref{trsf1,trsf4}. Since removing edges can only add conditional independences, it can only be worth doing if it helps in applying \cref{trsf4}. Since we are aiming for a GDAG with no unobserved nodes, the only point in applying \cref{trsf4} to an unobserved node would be if it helped with a future application between observed nodes. Clearly, adding an edge from an observed node to an unobserved node cannot help. Let us consider a situation in which \cref{trsf4} can be used to add an edge from an unobserved node to an observed node.  Now, we can (and will) later remove any such edge from unobserved nodes to observed nodes, except if it is required to apply \cref{trsf4}.  Because of this, the only point of adding the edge would be for it to connect the observed node to the one unobserved parent required to enable this later application of \cref{trsf4}. But, from the maximal application of \cref{trsf3}, any such role can just as easily be played by the unobserved parent required for the possible application of \cref{trsf4} presently under consideration.

Hence \Cref{trsf4} is only worth applying between observed nodes, and of these only the nodes in $T$ are possibilities. After all the transformation we are left with some GDAG, which defines a partial order on $T$ and can be extended to a total order. If we are aiming for a particular order we need to remove any edges from $T_j$ to $T_i$ with $j > i$. The only ultimate use for edges from unobserved nodes is to allow the application of \cref{trsf4}, for which only one such edge is needed. If an unobserved node has unobserved parents then by the first step the parent can only have more descendants, making it the same or more useful for the application of \cref{trsf4}. Hence the single edge from an unobserved node we keep might as well be from an element of $R$. 

Finally, we need to argue that \cref{trsf4} might as well be applied based on the ordering on $T$ we have defined using the final GDAG. The only point in applying a \cref{trsf4} early is if it helps with a later application of \cref{trsf4}. If the later application is to add an edge from $X$ to $Y$, we can only help by adding an edge from  a node in $\PA{X}$ to $Y$. But a node in $\PA{X}$ will be before $X$ in the ordering on $T$, so such an edge will, if possible, be added before when following the above strategy.

\section{GDAG reduction}\label{reduction}
In order to study whether or not the sufficient condition for $\Cs = \Is$ given in \cref{sufficientboring} might also be necessary for $\Cs = \Is$ by checking small GDAGs, it is useful to have a notion of when one GDAG ``reduces'' to another, such that if the second GDAG has $\Cs \subsetneq \Is$ then the first does as well.

\subsection{Strong reducibility}\label{strongred}
We say that a GDAG $A$ is strongly reducible to another GDAG $B$ if the observed nodes in $B$ are a subset of the observed nodes in $A$, and for any causal operational-probabilistic theory, the set of possible distributions on observable variables in $B$ is equal to the set of distributions obtained by marginalising distributions on $A$. We likewise require that $\Is$ for $B$ is exactly the marginals of $\Is$ for $A$.

Applying the following transformations to $A$ gives a new graph $B$ to which $A$ is strongly reducible:
\begin{enumerate}
  \item \emph{Removing a disconnected component.}\label{disconnect} By the definition of an operational-probabilistic theory, the probabilities for two disconnected components are the products of the probabilities for each. Since marginalising one factor in a product distribution gives the other factor, valid distributions for $A$ marginalise to valid distributions on $B$. Similarly a valid distribution on $B$ can be taken to a valid distribution on $A$ with the correct marginal by putting an arbitrary model on the removed component. For $\Is$ simply note that the $d$-separation conditions for $B$ are not affected by the presence or absence the disconnected component.
  \item \emph{Removing a childless unobserved node.} Such a node represents the unique deterministic effect, which can be factorised into the deterministic effect on each incoming system, which can then be incorporated into the definition of the parent node. To go in the other direction simply use trivial systems for each incoming edge. Such nodes can neither block existing paths nor create a new unblocked path and so do not effect $\Is$ either.

  \item \emph{Merging an unobserved node with its sole parent, also unobserved.} An unobserved node that has only one parent, which is also unobserved, represents a deterministic test. Its parent can be redefined by applying that test to the relevant output system. To go in the other direction simply use the identity test, whose existence is part of the definition of operational-probabilistic theory. This transformation also does not affect conditional independences among the observed nodes.

  \item \emph{Removing an observed node associated with a 1-outcome variable.}\label{1outvar} Such a node represents the deterministic effect on its inputs, as in the case of a childless unobserved node. Outgoing edges have no effect because they just add a fixed label to children. Finally, removing a node certainly cannot remove conditional independences from the remaining nodes, to ensure it doesn't add any see \cref{nonew}.

  \item \emph{For an observed node $X$ all of whose parents are observed, removing an edge from a parent $Y$ such that all the observable conditional independences from $d$-separation after the removal already held beforehand.}\label{uselessedge} Such an observed node is specified by a classical conditional probability $p(x|y,z)$. Once the edge from $Y$ is removed we have $\ci{X}{Y}{Z}$ (since if $Y$ is a descendant of $X$ the original graph would have contained a cycle). By assumption $\ci{X}{Y}{Z}$ therefore holds in the original distribution, i.e. $p(x|y,z) = p(x|z)$, and so we can achieve the exact same probability distributions with or without the edge from $Y$ to $X$. Finally, $\Is$ is the same by construction.
  \item \emph{Removing an unobserved node whose parents and children are subsets of the parents and children respectively of another unobserved node.} The test at such an unobserved node can simply be incorporated into the other node, with the edges from common parents and children now carrying the systems to/from both. To go in the other direction just add a trivial test to the new node. An unblocked path via the removed node can just as easily go via the other node so $\Is$ is unaffected.
\end{enumerate}

\subsection{Reducibility}

The condition for reducibility is the same as strong reducibility, except that we only consider generalised probabilistic theories that have system types, states, and measurements suitable for perfectly transmitting, encoding, and decoding any finite-valued classical information. This includes classical probability theory (which defines $\Cs$) and quantum theory (which defines $\Qs$). It also includes unspecified theories (which define $\Gs$) since any operational-probabilistic theory can always be supplemented with such systems. It does not include, for example, the restriction of quantum theory to operations with a certain amount of noise.

Clearly reducibility is a weaker notion than strong reducibility.  In addition to the transformations in the previous subsection, applying the following transformations to $A$ gives a new graph $B$ to which $A$ is reducible:
\begin{enumerate}
  \item \emph{Merging an unobserved node with its sole child.} To convert a model on the unmerged GDAG to the merged one, simply compose the two tests. To go in the other direction, let the new unobserved node with only one child be the identity test on the edges from unobserved parents, and use classical information encoding states for the incoming edges from observed parents. At the child use the corresponding classical information decoding measurements to recreate the correct dependencies. As for $\Is$, simply note that this change has no effect on the $d$-separation of observed nodes.
\item \emph{Merging an observed node $Y$ (that has only one sibling, $Z$) with its unobserved parent $X$ (which is itself parentless).} The pair of nodes $X,Y$ represents a bipartite state at $X$ with a measurement $Y$ on one system. Considered together this is a ``preparation test'' for the remaining system that goes to $Z$. But in a causal theory every state is proportional to a deterministic state, so this is equivalent to sampling from the classical probability distribution given by the norms of the states and then preparing the corresponding normalized state. The sampling can be done as the new consolidated node, whilst the preparation can be incorporated into $Z$. Going in the other direction, we are starting with a single node representing a classical probability distribution. This can be sampled as part of the new unobserved node $X$, with the resulting classical information transmitted to both children. The copy sent to the observed node $Y$ is simply decoded and output, the copy sent to the other node $Z$ is decoded and then used as the label that previously came from the observed parent. Since an unobserved node cannot be conditioned on, the path from the observed node $Y$ via the unobserved node $X$ operates in exactly the same way as a direct connection as far as $d$-separation is concerned, so $\Is$ is unchanged.
\end{enumerate}

\subsection{The implications of reducibility}
Suppose we have two GDAGs, and the first is reducible to the second. Suppose the second has $\Cs \subsetneq \Is$, i.e. there exists some $P \in \Is$ with $P \not\in \Cs$. Then by reducibility, there exists a $P' \in \Is$ for the first GDAG, which marginalises to $P$. Suppose $P' \in \Cs$ for the first GDAG. Then by a second application of reducibility, it marginalises to a distribution in $\Cs$ for the second GDAG. But we already said it marginalises to $P \not\in \Cs$. Hence $P' \not\in\Cs$. We conclude that if a GDAG has $\Cs \subsetneq \Is$ then so does any other GDAG that reduces to it.

Except for \cref{disconnect,1outvar,uselessedge} of \cref{strongred}, the reduction rules don't affect the observed nodes and so the marginalisation step in the definition of reducibility is irrelevant. For reductions that don't use those 3 rules, we therefore have the stronger statement that $\Cs$ is the same for both GDAGs, and so is $\Is$. In particular $\Cs \subsetneq \Is$ for one GDAG if and only if $\Cs \subsetneq \Is$ for the other.

\subsection{$d$-separation without a trivial variable}\label{nonew}
The following is needed to ensure that transformation \cref{1outvar} of section \ref{strongred} satisfies the part of the definition of strong reducibility relating to $\Is$. Given a GDAG with observed and unobserved nodes, suppose that a distribution $P$ over variables on the observed nodes satisfies all the conditional independences implied by $d$-separation, i.e. $P \in \Is$. Suppose further that some variables $F$ always takes a fixed value. Then we claim that $P$ also satisfies all the conditional independences implied by the GDAG with $F$ removed.

Suppose that $X$ and $Y$ are $d$-separated by $Z$ in the new GDAG but not the old. If we imagine removing the edges incident to $F$ one by one, starting with outgoing edges and then moving on to incoming edges, then there must be a ``critical edge'' wherein $X$ and $Y$ are not $d$-separated by $Z$ before the removal, but are $d$-separated afterwards. Therefore all the unblocked paths before the removal must have passed through the critical edge.

Consider first an outgoing critical edge. Then $X$ and $Y$ are $d$-separated by $ZF$, because $F$ blocks any otherwise unblocked path from $X$ to $Y$. That means that $\ci{X}{Y}{ZF}$. But if $F$ takes a fixed value then conditioning on it doesn't do anything, so $\ci{X}{Y}{Z}$ as required.

The other case is an incoming critical edge. By construction all the outgoing edges have already been removed, so all the unblocked paths from $X$ to $Y$ are head-to-head at $F$. If we write $Z = Z_\text{D}Z_\text{ND}$ where $Z_\text{D}$ are descendants of $F$ and $Z_\text{ND}$ are not, then $X$ and $Y$ are $d$-separated by $Z_\text{ND}$ and so $\ci{X}{Y}{Z_\text{ND}}$. Furthermore any path from $XY$ to $Z_\text{D}$ not blocked by $Z_\text{ND}$ passes through $F$, and so $\ci{Z_\text{D}}{XY}{Z_\text{ND}F}$. As before this implies that $\ci{Z_\text{D}}{XY}{Z_\text{ND}}$. By the decomposition property of conditional independences we have $\ci{Z_\text{D}}{X}{YZ_\text{ND}}$ and hence $\ci{X}{Z_\text{D}}{YZ_\text{ND}}$ by the symmetry property. Combining this with $\ci{X}{Y}{Z_\text{ND}}$ using the contraction property gives $\ci{X}{Y}{Z_\text{D}Z_\text{ND}} = \ci{X}{Y}{Z}$ as required.

\section{Small ``interesting'' GDAGs}\label{dagtable}
Here we present the all the GDAGs of size at most six which the criteria in \cref{sec:towards} does not identify as having $\Cs = \Is$, and the reduction criteria above do not identify as being reducible to a smaller such GDAG. If all these GDAGs have $\Cs \subsetneq \Is$ then our criteria is also necessary for $\Cs = \Is$, at least for GDAGs of this size.

As well as the GDAG itself, we list a generating set of observable independences, which defines $\Is$. We also list a generating set of Shannon-type inequalities for $\Cs$, excluding those that are Shannon-type inequalities for $\Is$.

These inequalities provide good evidence that $\Cs \subsetneq \Is$. However, technically these inequalities could be non-Shannon inequalities for $\Is$. For most of the GDAGs, we have highlighted a subset of the nodes. The probability distribution defined by perfectly correlated random bits on these nodes, with all other nodes taking a fixed value, is a member of $\Is$ yet violates the first entropic inequality listed and is therefore not in $\Cs$. This closes the non-Shannon ``loophole'' and establishes that $\Cs \subsetneq \Is$ for these GDAGs.

\input{pics/dagtable}
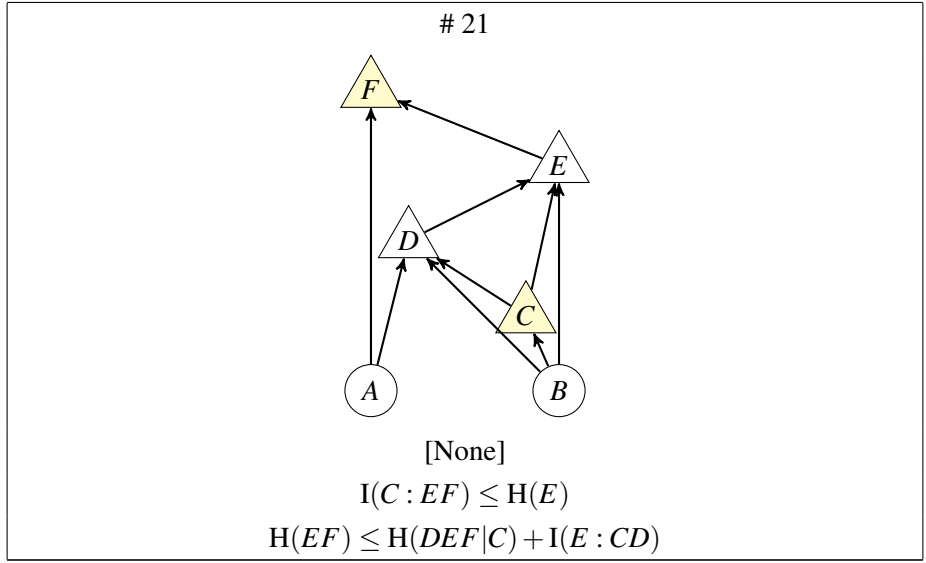
\end{document}